\documentclass[11pt, reqno]{article}
\usepackage[hmargin=1.0in,vmargin=1.0in]{geometry}
\usepackage{amsthm}
\usepackage[english]{babel}
\usepackage{cite}
\usepackage{color}
\usepackage{amsmath,amsfonts,amssymb,graphicx}
\usepackage{microtype}
\usepackage{mathrsfs,url}
\usepackage{enumerate}
\usepackage{cmathml}
\usepackage{caption}

\DeclareMathOperator{\proj}{\ensuremath{e}}
\DeclareMathOperator{\argmin}{argmin}
\DeclareMathOperator{\wops}{{\bf W}}
\DeclareMathOperator{\wpol}{wPol}

\newcommand{\inst}{\ensuremath{I}}

\DeclareMathOperator{\rel}{{\bf R}}
\newcommand{\tup}[1]{\ensuremath{\mathbf #1}}
\newcommand{\tuple}[1]{\ensuremath{(#1)}}
\newcommand{\ignore}[1]{}

\DeclareMathOperator{\wrel}{{\bf \Phi}}
\DeclareMathOperator{\supp}{{\rm supp}}
\DeclareMathOperator{\VCSP}{VCSP}

\DeclareMathOperator{\feas}{Feas}
\DeclareMathOperator{\perm}{Cycl}
\DeclareMathOperator{\dom}{{Feas}}

\DeclareMathOperator{\wclone}{wClone}
\DeclareMathOperator{\imp}{Imp}
\DeclareMathOperator{\inv}{Inv}

\DeclareMathOperator{\pol}{Pol}
\DeclareMathOperator{\wrelclone}{wRelClone}
\DeclareMathOperator{\relclone}{RelClone}
\newcommand{\Q}{\mbox{$\mathbb Q$}}
\newcommand{\qq}{\ensuremath{\overline{\mathbb{Q}}}}

\date{}

\newtheorem{theorem}{Theorem}
\newtheorem{lemma}{Lemma} 
\newtheorem*{lemma*}{Lemma} 
\newtheorem*{proposition*}{Proposition} 
\newtheorem*{theorem*}{Theorem} 
\newtheorem{proposition}[theorem]{Proposition}
\newtheorem{corollary}{Corollary}
\newtheorem{definition}{Definition}
\theoremstyle{remark}
\newtheorem{example}{Example}

\begin{document}
\title{Necessary Conditions for Tractability of Valued CSPs\thanks{Stanislav \v{Z}ivn\'y was supported by a Royal Society University Research Fellowship.}}

\author{
Johan Thapper\\
Universit\'e Paris-Est, Marne-la-Vall\'ee, France\\
\texttt{thapper@u-pem.fr}
\and
Stanislav \v{Z}ivn\'{y}\\
University of Oxford, UK\\
\texttt{standa.zivny@cs.ox.ac.uk}
}

\maketitle

\begin{abstract}

The connection between constraint languages and clone theory has been a fruitful
line of research on the complexity of constraint satisfaction problems. In a
recent result, Cohen et al.~[SICOMP'13] have characterised a Galois connection
between valued constraint languages and so-called weighted clones. In this
paper, we study the structure of weighted clones. We extend the results of Creed
and \v{Z}ivn\'y from~[CP'11/SICOMP'13] on types of weightings necessarily
contained in every nontrivial weighted clone. This result has immediate
computational complexity consequences as it provides necessary conditions for
tractability of weighted clones and thus valued constraint languages. 
We demonstrate that some of the necessary conditions are
also sufficient for tractability, while others are provably not.

\end{abstract}

\section{Introduction}\label{sec:intro}

The constraint satisfaction problem (CSP) is a general framework capturing
decision problems arising in many contexts of computer
science~\cite{Dechter03:processing,Apt03:principles,Hell08:survey}. The CSP is NP-hard in general but there has been
much success in finding tractable fragments of the CSP by restricting the types
of relation allowed in the constraints.  A set of allowed relations has been
called a \emph{constraint language}~\cite{Feder98:monotone,Jeavons98:algebraic}.
For some constraint languages the associated constraint satisfaction problems
with constraints chosen from that language are solvable in polynomial-time,
whilst for other constraint languages this class of problems is
NP-hard~\cite{Jeavons97:closure,Feder98:monotone}; these are referred to as
\emph{tractable languages} and \emph{NP-hard languages}, respectively. Dichotomy
theorems, which classify each possible constraint language as either tractable
or NP-hard, have been established for constraint languages over two-element
domains~\cite{Schaefer78:complexity}, three-element
domains~\cite{Bulatov06:3-elementJACM}, for conservative (containing all unary
relations) constraint
languages~\cite{Bulatov11:conservative}, for maximal constraint
languages~\cite{Bulatov01:complexity,Bulatov04:graph}, for graphs (corresponding
to languages containing a single binary symmetric relation)~\cite{Hell90:h-coloring}, and for
digraphs without sources and sinks (corresponding to languages containing a
single binary relations without sources and sinks)~\cite{Barto09:siam}. The most successful
approach to classifying the complexity of constraint languages has been the
algebraic approach~\cite{Jeavons97:closure,Bulatov05:classifying,Barto14:jacm}.

The \emph{valued} constraint satisfaction problem (VCSP) is a generalisation of
the CSP that captures not only decision problems but also optimisation
problems~\cite{Cohen06:complexitysoft,z12:complexity,jkz14:beatcs}. A VCSP
instance associates with each constraint a \emph{weighted relation}, which is a
$\qq$-valued function, where $\qq=\mathbb{Q}\cup\{\infty\}$ is the set of
extended rational numbers, and the goal is to minimise the sum of the weighted
relations associated with all constraints. Tractable fragments of the VCSP have
been identified by restricting the types of allowed weighted relations that can
be used to define the valued constraints. A set of allowed weighted relations
has been called a \emph{valued constraint
language}~\cite{Cohen06:complexitysoft}. Classifying the complexity of
\emph{all} valued constraint languages is a challenging task as it includes as a
special case the classification of $\{0,\infty\}$-valued languages (i.e.,
constraint languages), which would answer the conjecture of Feder and
Vardi~\cite{Feder98:monotone}, which asserts that every constraint language is
either tractable or NP-hard, and its algebraic refinement, which specifies the
precise boundary between tractable and NP-hard
languages~\cite{Bulatov05:classifying}. However, several nontrivial results are
known. Dichotomy theorems, which classify each possible valued constraint
language as either tractable or NP-hard, have been established for valued
constraint languages over two-element domains~\cite{Cohen06:complexitysoft}, for
conservative (containing all $\{0,1\}$-valued unary cost functions) valued
constraint languages~\cite{kz13:jacm}, and for minimum-solution
languages (containing relations and a single unary injective weighted relation)~\cite{tz15:icalp}. Furthermore, it has been shown that a dichotomy for
constraint languages implies a dichotomy for valued constraint
languages~\cite{Kolmogorov15:general-valued}. Moreover, the power of the basic
linear programming relaxation~\cite{tz12:focs,kolmogorov13:icalp,ktz15:sicomp}
and the power of the Sherali-Adams relaxations~\cite{tz15:icalp} for valued
constraint languages have been characterised.

Cohen et al. have recently introduced an algebraic theory of weighted
clones~\cite{cccjz13:sicomp} for classifying the computational complexity of
valued constraint languages. This
theory establishes a one-to-one correspondence between valued constraint
languages closed under expressibility (which does not change the complexity of
the associated class of optimisation problems), called weighted relational
clones, and weighted clones~\cite{cccjz13:sicomp}.
This is an extension of (part of) the algebraic approach to CSPs which relies on a
one-to-one correspondence between constraint languages closed under
pp-definability (which does not change the
complexity of the associated class of decision problems), called
relational clones, and clones~\cite{Bulatov05:classifying}, thus making it possible
to use deep results from universal algebra.

Creed and \v{Z}ivn\'y initiated the study of weighted clones and have used the theory of weighted clones to determine
certain necessary conditions on nontrivial weighted clones and thus on
tractable valued constraint languages~\cite{cz11:cp-mwc}, see
also~\cite{cccjz13:sicomp}. In particular,~\cite{cz11:cp-mwc} simplifies the
NP-hardness part of the complexity classification of Boolean valued constraint
languages from~\cite{Cohen06:complexitysoft}. 

\paragraph{Contributions}
We continue the study of weighted clones started
in~\cite{cz11:cp-mwc,cccjz13:sicomp}. After introducing valued constraint
satisfaction problems and all necessary tools in Section~\ref{sec:prelim}, we
study, in Section~\ref{sec:wclones}, structural properties of nontrivial
weighted clones.  
Our main result on weighted clones, Theorem~\ref{thm:class2}, is an extension of
a result from~\cite{cccjz13:sicomp} that provides a more fine-grained
characterisation of what conditions on weighted clones are necessary for
tractability. Moreover, we demonstrate that some of the necessary conditions are
also sufficient for tractability, while others are provably not. Overall, we
give a structural result that shows what types of weightings are guaranteed to
exist in nontrivial weighted clones. As a direct consequence, we narrow down the
possible structure of tractable weighted clones.
A proof of our main result is presented in Section~\ref{sec:proof} and is based
on an application of Gordan's theorem, which is a variant of LP duality. The
introduced technique is novel and might prove useful in future work on weighted
clones. 
Finally, we relate our results to \emph{maximal} tractable valued constraint
languages, or equivalently, to \emph{minimal} tractable weighted clones.
\paragraph{Related work}
Given the generality of the VCSP, there have been results on the complexity of
special types of VCSPs. Finite-valued CSPs are VCSPs in which all weighted
relations are $\mathbb{Q}$-valued. In other words, finite-valued CSPs are purely
optimisation problems and thus do not include as a special case (decision) CSPs.
The authors have
recently classified all finite-valued constraint languages on arbitrary finite
domains~\cite{tz13:stoc}. Minimum Solution (Min-Sol) problems are Valued CSPs
with one unary \emph{injective} $\mathbb{Q}$-valued weighted relation and
$\{0,\infty\}$-valued weighted relations. Min-Sols generalise
Min-Ones~\cite{Creignouetal:siam01} and bounded integer linear programs.
Min-Sols have been only very recently classified~\cite{tz15:icalp} with respect
to computational complexity, thus improving on previous partial
classifications~\cite{Khanna01:approximability,Uppman13:icalp,Jonsson07:maxsol,Jonsson08:max-sol,Jonsson08:siam}.
Minimum Cost Homomorphism (Min-Cost-Hom) problems are Valued CSPs in which all
but unary weighted relations are $\{0,\infty\}$-valued. Thus the optimisation
part of the problem is only given by a sum of unary terms. This may seem very
restrictive but it is known~\cite{ccjz15:aaai,Powell15:arxiv} 
that any VCSP is equivalent to a VCSP where only the
(not necessarily injective) unary constraints involve optimisation. Min-Cost-Hom
problems with all unary cost functions have been classified
in~\cite{Takhanov10:stacs}. Also, Min-Cost-Hom problems with all unary
$\{0,\infty\}$-valued cost functions~\cite{Takhanov10:cocoon,Uppman13:icalp} and
on three-element domains~\cite{Uppman14:stacs} have been classified.

\section{Preliminaries}
\label{sec:prelim}

\subsection{Valued CSPs}

Throughout the paper, let $D$ be a fixed finite set of size at least two.

\begin{definition}\label{def:rel}
An $m$-ary \emph{relation} over $D$ is any mapping $\phi:D^m\to\{c,\infty\}$ for
some $c\in\mathbb{Q}$.
We denote by $\rel_D^{(m)}$ the set of all $m$-ary relations and let
$\rel_D=\bigcup_{m\geq 1}{\rel_D^{(m)}}$.
\end{definition}

An $m$-ary relation over $D$ is commonly defined as a subset of $D^m$. Note that
Definition~\ref{def:rel} is equivalent to the standard definition as any subset
of $D^m$ can be associated with the set $\{\tup{x}\in
D^m\:|\:\phi(\tup{x})<\infty\}$. Consequently, we shall use both definitions
interchangeably. 

Given an $m$-tuple $\tup{x}\in D^m$, we denote its $i$th entry by $\tup{x}[i]$ for $1\leq i\leq m$.

Let $\qq=\mathbb{Q}\cup\{\infty\}$ denote the set of rational numbers with (positive) infinity.

\begin{definition}\label{def:wrel}
An $m$-ary \emph{weighted relation}\footnote{In some paper weighted relations are called cost functions.} over $D$ is any mapping $\phi:D^m\to\qq$. 
We denote by $\wrel_D^{(m)}$ the set of all $m$-ary weighted relations and let
$\wrel_D=\bigcup_{m\ge 1}{\wrel_D^{(m)}}$. 
\end{definition}

From Definition~\ref{def:wrel} we have that relations are a special type of
weighted relations. 
If needed we call a weighted relation \emph{unweighted} to emphasise the fact that $\phi$ is a relation.

\begin{example}
An important example of a (weighted) relation is the binary equality $\phi_=$ on $D$:
$\phi_=(x,y)=0$ if $x=y$ and $\phi_=(x,y)=\infty$ if $x\neq y$.
\end{example}

For any $m$-ary weighted relation $\phi\in\wrel_D^{(m)}$, we denote by
$\feas(\phi)=\{\tup{x}\in D^m\:|\:\phi(\tup{x})<\infty\}\in\rel_D^{(m)}$ the
underlying \emph{feasibility relation}. 

A weighted relation $\phi:D^m\to\qq$ is called \emph{finite-valued} if
$\feas(\phi)=D^m$.

\begin{definition}
Let $V=\{x_1,\ldots, x_n\}$ be a set of variables. A \emph{valued constraint} over $V$ is an expression
of the form $\phi(\tup{x})$ where $\phi\in \wrel_D^{(m)}$ and $\tup{x}\in V^m$. The number $m$ is called the \emph{arity} of the constraint,
the weighted relation $\phi$ is called the \emph{constraint weighted relation},
and the tuple $\tup{x}$ the \emph{scope} of the constraint.
\end{definition}

We call $D$ the \emph{domain}, the elements of $D$ \emph{labels} (for variables), and say that the
weighted relation in $\wrel_D$ take \emph{values}.

\begin{definition}
An instance of the \emph{valued constraint satisfaction problem} (VCSP) is specified
by a finite set $V=\{x_1,\ldots,x_n\}$ of variables, a finite set $D$ of labels,
and an \emph{objective function} $\inst$
expressed as follows:
\begin{equation}
\inst(x_1,\ldots, x_n)=\sum_{i=1}^q{\phi_i(\tup{x}_i)}\,,
\label{eq:sepfun}
\end{equation}
where each $\phi_i(\tup{x}_i)$, $1\le i\le q$, is a valued constraint over $V$.
Each constraint can appear multiple times in $\inst$.

The goal is to find an \emph{assignment} (or a \emph{labelling}) of labels to the variables that minimises $\inst$.
\end{definition}

CSPs are a special case of VCSPs with (unweighted) relations with the goal to
determine the existence of a feasible solution.
\begin{example} \label{ex:maxcut} The \textsc{Max-Cut} problem for a graph is to
find a cut with the largest possible size. This problem is
NP-hard~\cite{Garey79:intractability}
and equivalent to the \textsc{Min-UnCut} problem
with respect to exact solvability.
For a graph $(V,E)$ with $V=\{x_1,\ldots,x_n\}$, this problem can be expressed
as the VCSP instance $\inst(x_1,\ldots,x_n)=\sum_{(i,j)\in E} \phi_{\sf
xor}(x_i,x_j)$ over the Boolean domain $D=\{0,1\}$, where $\phi_{\sf
xor}:\{0,1\}^2\to\qq$ is defined by $\phi_{\sf
xor}(x,y)=1$ if $x=y$ and $\phi_{\sf xor}(x,y)=0$ if $x\neq y$.
\end{example}

\begin{definition}
Any set $\Gamma\subseteq\wrel_D$ is called a \emph{valued constraint
language}\footnote{A valued constraint language $\Gamma$ is sometimes called
\emph{general-valued} to emphasise the fact that weighted relations from
$\Gamma$ are not necessarily finite-valued.} over $D$, or simply a \emph{language}. We
will denote by $\VCSP(\Gamma)$ the class of all VCSP instances in which the
constraint weighted relations are all contained in $\Gamma$.
\end{definition}

\begin{definition}
A valued constraint language $\Gamma$ is called \emph{tractable} if
$\VCSP(\Gamma')$ can be solved (to optimality) in
polynomial time for every finite subset $\Gamma'\subseteq\Gamma$, and $\Gamma$
is called \emph{intractable} if $\VCSP(\Gamma')$ is NP-hard for some finite
$\Gamma'\subseteq\Gamma$.
\end{definition}

A valued constraint language is called \emph{finite-valued} if every weighted
relation $\phi$ from the language is finite-valued. 
Example~\ref{ex:maxcut} shows that the finite-valued constraint language
$\{\phi_{\sf xor}\}$ is intractable. 

We denote by $\feas(\Gamma)=\{\dom(\phi)\:|\:\phi\in\Gamma\}$ the set of
underlying relations of all weighted relations from $\Gamma$.

\subsection{Weighted relational clones}
\begin{definition} \label{def:expres}
We say that an $m$-ary weighted relation $\phi$ is \emph{expressible} over a valued constraint
language $\Gamma$ if there exists a VCSP instance $\inst \in \VCSP(\Gamma)$ 
with variables $V=\{x_1,\ldots,x_n,y_1,\ldots,y_m\}$,
such that
\begin{equation}
\phi(y_1,\ldots,y_m) = \min_{x_1\in D,\ldots,x_n\in D}\inst(x_1,\ldots,x_n,y_1,\ldots,y_m)\,.
\end{equation}
\end{definition}
A valued constraint language $\Gamma$ is \emph{closed under expressibility} if
every weighted relation $\phi$ expressible over $\Gamma$ belongs to $\Gamma$.
\begin{definition}
\label{defn:wrelclone}
A valued constraint language $\Gamma \subseteq \wrel_D$ is called a \emph{weighted relational clone} 
if it contains the binary equality relation $\phi_=$ on $D$ and is closed under
expressibility, scaling by non-negative rational constants (where we define
$0\cdot\infty=\infty$),
and addition of rational constants.

For any $\Gamma$, we define $\wrelclone(\Gamma)$ to be the smallest weighted relational clone containing $\Gamma$.
\end{definition}

Note that for any weighted relational clone $\Gamma$ if $\phi\in\Gamma$ then
$\feas(\phi)\in\Gamma$ as $\feas(\phi)=0\phi$. 

\begin{definition}
A \emph{relational clone} is a weighted
relational clone containing only (unweighted) relations.\footnote{Equivalently,
a set of relations containing the binary equality relation and closed under
conjunction and existential quantification.} For a set of relations $\Gamma$, we
denote by $\relclone(\Gamma)$ the smallest relational clone containing $\Gamma$.
\end{definition}

It has been shown that $\Gamma$ is tractable if and only if $\wrelclone(\Gamma)$
is tractable~\cite{cccjz13:sicomp}. Consequently, when trying to identify
tractable valued constraint languages, it is sufficient to consider only
weighted relational clones. 

\subsection{Weighted clones}

Any mapping $f:D^k\rightarrow D$ is called a $k$-ary \emph{operation}. We will
apply a $k$-ary operation $f$ to $k$ $m$-tuples $\tup{x_1},\ldots,\tup{x_k}\in
D^m$ coordinatewise, that is, 
\begin{equation}
f(\tup{x_1},\ldots,\tup{x_k})=(f(\tup{x_1}[1],\ldots,\tup{x_k}[1]),\ldots,f(\tup{x_1}[m],\ldots,\tup{x_k}[m]))\,.
\end{equation}
\begin{definition} \label{def:pol}
Let $\phi$ be an $m$-ary weighted relation on $D$ and let $f$ be a $k$-ary operation
on $D$.
Then $f$ is a \emph{polymorphism} of $\phi$ if,
for any $\tup{x_1},\tup{x_2},\ldots,\tup{x_k} \in D^m$ with
$\tup{x_i}\in\dom(\phi)$ for all $1\leq i\leq k$,
we have that $f(\tup{x_1},\tup{x_2},\ldots,\tup{x_k})\in\dom(\phi)$.

For any valued constraint language $\Gamma$ over a set $D$,
we denote by $\pol(\Gamma)$ the set of all operations on $D$ which are polymorphisms of all 
$\phi \in \Gamma$. We write $\pol(\phi)$ for $\pol(\{\phi\})$.
\end{definition}
A $k$-ary \emph{projection} is an operation of the form
$\proj^{(k)}_i(x_1,\ldots,x_k)=x_i$ for some $1\leq i\leq k$.
Projections are (trivial) polymorphisms of all valued constraint languages.

\begin{definition}
The \emph{superposition} of a $k$-ary operation $f:D^k\rightarrow D$ with $k$
$\ell$-ary operations $g_i:D^\ell\rightarrow D$ for $1\leq i\leq k$ is the
$\ell$-ary function $f[g_1,\ldots,g_k]:D^\ell\to D$ defined by
\begin{equation}
f[g_1,\ldots,g_k](x_1,\ldots,x_\ell)=f(g_1(x_1,\ldots,x_\ell),\ldots,g_k(x_1,\ldots,x_\ell))\,.
\end{equation}
\end{definition}

\begin{definition}\label{def:clone}
A \emph{clone} of operations, $C$, is a set of operations on $D$
that contains all projections and is closed under superposition.
The $k$-ary operations in a clone $C$ will be denoted by $C^{(k)}$.
\end{definition}

\begin{example}\label{ex:clone}
For any $D$, let $\mathbf{J}_D$ be the set of all projections on $D$ and
$\mathcal{O}_D$ be the set of all operations on $D$. By
Definition~\ref{def:clone}, both $\mathbf{J}_D$ and $\mathcal{O}_D$ are clones.
\end{example}

It is well known that $\pol(\Gamma)$ is a clone for all valued constraint
languages $\Gamma$~\cite{Denecke-Wismath02}.

\begin{definition} \label{defn:wop}
A $k$-ary \emph{weighting}
of a clone $C$ is a function $\omega : C^{(k)} \rightarrow \mathbb{Q}$ such that
$\omega(f) < 0$ only if $f$ is a projection and
\begin{equation}\label{ineq:weightingzero}
  \sum_{f \in C^{(k)}}\omega(f)\ =\ 0\,.
\end{equation}
We define $\supp(\omega)=\{f\in C^{(k)}\:|\:\omega(f)>0\}$.
\end{definition}
\begin{definition} \label{defn:wp_trans}
For any clone $C$, any $k$-ary weighting $\omega$ of $C$, and any $g_1,\ldots,g_k \in
C^{(\ell)}$, the \emph{superposition}
of $\omega$
and $g_1,\ldots,g_k$, is the function $\omega[g_1,\ldots,g_k]: C^{(\ell)}
\rightarrow \mathbb{Q}$ defined by
\begin{equation}
\omega[g_1,\ldots,g_k](f') = \sum_{\{f \in C^{(k)}\:|\:f[g_1,\ldots,g_k] =
f'\}}\omega(f)\,.
\end{equation}
If $\omega$ satisfies~\eqref{ineq:weightingzero} then so does
$\omega[g_1,\ldots,g_k]$.
If the result of a superposition is a valid weighting (that is, negative weights
are only assigned to projections) then that superposition will be called a
\emph{proper} superposition.
\end{definition}
We remark that the superposition (of an operation with other operations) is also
known as composition. On the other hand, the superposition of a $k$-ary
weighting $\omega$ with $k$ $\ell$-ary operations $g_1,\ldots,g_k$ can be seen as multiplying $\omega$, seen as
a (row) vector, by a matrix with rows indexed by $k$-ary operations and columns
indexed by $\ell$-ary operations; given a row operation $f$ and a column
operation $f'$ the corresponding entry in the matrix is 1 if
$f[g_1,\ldots,g_k]=f'$ and 0 otherwise. The result of this matrix multiplication is a vector of weights
assigned to $\ell$-ary operations.
\begin{definition} \label{defn:wclone}
A \emph{weighted clone}, $W$, is a non-empty set of weightings of some fixed
clone $C$, called the \emph{support clone} of $W$, which is closed under nonnegative
scaling, addition of weightings of equal arity, and proper superposition with
operations from $C$. We define $\supp(W)=\bigcup_{\omega\in W}\supp(\omega)$.
\end{definition}

\begin{example}\label{ex:wclones}
Let $C$ be a clone. We give examples of two weighted clones with support
clone $C$.

\begin{enumerate}
\item $\mathbf{W}_C^0$ is the zero-valued weighted clone, that is, the
weighted clone containing, for each arity $k$, a weighting
$\omega_k\in\mathbf{W}_C^0$ with
$\omega_k(f)=0$ for all $f\in C^{(k)}$.

\item $\mathbf{W}_C$ is the weighted clone containing all possible weightings of $C$.
\end{enumerate}

By Definition~\ref{ex:wclones}, weighted clones are closed under nonnegative
scaling. Consequently, by scaling by zero, any weighted clone
$W$ with support clone $C$ contains $\mathbf{W}_C^0$, which is the
inclusion-wise smallest weighted clone with support clone $C$. On the other
hand, $\mathbf{W}_C$ is the inclusion-wise largest weighted clone with support
clone $C$.
\end{example}

\begin{example}
It is easy to show that $\supp(W)\cup\mathbf{J}_D$ is a clone for any weighted clone
$W$ defined on $D$~\cite{cccjz13:sicomp}, see
also~\cite{Kozik15:icalp,tz15:icalp}.
\end{example}

We now establish a correspondence between weightings and weighted relations, which will
allow us to link weighted clones and weighted relational clones.
\begin{definition} \label{def:wp} 
Let $\phi$ be an $m$-ary weighted relation 
on $D$ and let $\omega$ be a $k$-ary weighting of a clone $C$ of
operations on $D$.
We call $\omega$ a \emph{weighted polymorphism}
of $\phi$ if $C\subseteq\pol(\phi)$ and for any
$\tup{x}_1,\tup{x}_2,\ldots,\tup{x}_k \in D^m$ with
$\tup{x}_i\in\dom(\phi)$ 
for all $1\leq i\leq k$, we have 
\begin{equation}\label{eq:WPOL}
  \sum_{f \in C^{(k)}}\omega(f)\phi(f(\tup{x}_1,\tup{x}_2,\ldots,\tup{x}_k)) \ \leq\ 0\,.
\end{equation}
If $\omega$ is a weighted polymorphism of $\phi$ we say that $\phi$ is \emph{improved} by $\omega$.
\end{definition}
\begin{example}
Consider the class of submodular functions~\cite{Nemhauser88:optimization}.
These are precisely the functions $\phi$ defined on $D=\{0,1\}$ satisfying
$\phi(\min(\tup{x}_1,\tup{x}_2)) + \phi(\max(\tup{x}_1,\tup{x}_2)) -
\phi(\tup{x}_1) - \phi(\tup{x}_2) \leq 0$, where $\min$ and $\max$ are the two
binary operations that return the smaller and larger of its two arguments
respectively with respect to the usual order $0<1$. In other words, the set of submodular
functions is the set of weighted relations with a binary weighted polymorphism
$\omega_{sub}$ defined by: $\omega_{sub}(f)=-1$ if $f \in
\{e_1^{(2)},e_2^{(2)}\}$, $\omega_{sub}(f)=+1$ if $f \in
\{\min,\max\}$, and $\omega_{sub}(f)=0$ otherwise.
\end{example}

\begin{definition}
For any $\Gamma \subseteq \wrel_D$, we define $\wpol(\Gamma)$ to be the set of all
weightings of $\pol(\Gamma)$ which are weighted polymorphisms of all weighted
relations $\phi \in \Gamma$. We write $\wpol(\phi)$ for $\wpol(\{\phi\})$.
\end{definition}

\begin{definition}
We define $\wops_D$ to be the union of the sets $\wops_C$ over all clones $C$ on $D$.
\end{definition}

Any $W \subseteq \wops_D$ may contain weightings of \emph{different}
clones over $D$. We can then  extend each of these weightings with zeros, as
necessary, so that they are weightings of the same clone $C$, where $C$ is the
smallest clone containing all the clones associated with weightings in $W$.

\begin{definition}
We define $\wclone(W)$ to be the smallest weighted clone
containing this set of extended weightings obtained from $W$.
\end{definition}

For any $W \subseteq \wops_D$,
we denote by $\imp(W)$ the set of all weighted relations in $\wrel_D$ which are
improved by all weightings $\omega \in W$.

\begin{example}
Every weighting in $\mathbf{W}_{\mathbf{J}_D}^0$ is a weighted polymorphism of
any possible weighted relation. Hence $\imp(\mathbf{W}_{\mathbf{J}_D}^0)=\wrel_D$.

The weighted relations that are improved by all weightings are precisely those
which take at most one value. Hence $\imp(\mathbf{W}_{\mathbf{J}_D})=\rel_D$.
\end{example}

\begin{definition}
A weighted clone $W$ is called \emph{tractable} if $\imp(W)$ is tractable, and
intractable if $\imp(W)$ is intractable.
\end{definition}

The main result in~\cite{cccjz13:sicomp} establishes a 1-to-1 
correspondence between weighted relational clones and weighted clones.
\begin{theorem}[\hspace*{-0.3em}\cite{cccjz13:sicomp}]~\label{thm:wgc}
\begin{enumerate}
\item
For any finite $D$ and any finite $\Gamma \subseteq \wrel_D$, $\imp(\wpol(\Gamma)) = \wrelclone(\Gamma)$.
\item
For any finite $D$ and any finite $W\subseteq \wops_D$, $\wpol(\imp(W)) = \wclone(W)$.
\end{enumerate}
\end{theorem}
Thus, when trying to identify tractable valued constraint languages, it is
sufficient to consider only languages of the form $\imp(W)$ for some weighted clone $W$.

\begin{definition}
A weighting is called \emph{positive} if it assigns positive weight to at least
one operation that is not a projection. 
\end{definition}

Positive weightings are necessary for tractability: any tractable weighted clone
$W$ contains a positive weighting~\cite[Corollary
7.4]{cccjz13:sicomp}. Consequently, throughout this paper we will be only concerned with weighted
clones that contain a positive weighting.

\subsection{Properties of operations}

We finish this section with a discussion of certain types of operations.
For any $k\geq 2$, a $k$-ary operation $f$ is called \emph{sharp} if $f$ is not
a projection, but the operation obtained by equating any two inputs in $f$ is a
projection~\cite{Csakany05:minimal}.  All sharp operations must satisfy the
identity $f(x,x,\ldots,x) = x$; such operations are called \emph{idempotent}.
Ternary sharp operations may be classified according to their labels on tuples
of the form $(x,x,y), (x,y,x)$ and $(y,x,x)$, which must be equal to either $x$
or $y$. There are precisely $8$ possibilities, as listed in
Table~\ref{tab:sharp3}.

  \begin{table}[t]
  \begin{center}
  \begin{tabular}{c|*{8}{c}}
    \makebox[1.5cm]{Input} 
    & \makebox[0.75cm]{Mj} & \makebox[0.75cm]{S1} & \makebox[0.75cm]{S2}
    & \makebox[0.75cm]{S3} & \makebox[0.75cm]{P1} & \makebox[0.75cm]{P2}
    & \makebox[0.75cm]{P3} & \makebox[0.75cm]{Mn}\\
    \hline
    (x,x,y) & x & x & x & y & x & y & y & y\\
    (x,y,x) & x & x & y & x & y & x & y & y\\
    (y,x,x) & x & y & x & x & y & y & x & y\\
  \end{tabular}
  \end{center}
  \caption{Sharp ternary operations}
  \label{tab:sharp3}
  \end{table}
The first column in Table~\ref{tab:sharp3} corresponds to operations that
satisfy the identities $f(x,x,y)=f(x,y,x)=f(y,x,x)=x$ for all $x,y\in D$; such
operations are called \emph{majority} operations. The last column in the table
corresponds to operations that satisfy the identities
$f(x,x,y)=f(x,y,x)=f(y,x,x)=y$ for all $x,y\in D$; such operations are called
\emph{minority} operations. Columns 5, 6, and 7 in Table~\ref{tab:sharp3}
correspond to operations that satisfy the identities $f(y,y,x) = f(x,y,x) =
f(y,x,x)=y$ for all $x,y \in D$ (up to permutations of inputs); such operations
are called \emph{Pixley} operations~\cite{Csakany05:minimal}. 
For any $k \geq 3$, a $k$-ary operation $f$ is called a \emph{semiprojection} if
it is not a projection, but there is an index $1\leq i \leq k$ such that
$f(x_1,\ldots,x_k)=\proj^{(k)}_i$ for all $x_1,\ldots,x_k\in D$ such that
$x_1,\ldots,x_k$ are not pairwise distinct. In other words, a
semiprojection is a particular form of sharp operation where the operation
obtained by equating any two inputs is always the \emph{same} projection.
Columns 2, 3, and 4 in Table~\ref{tab:sharp3} correspond to semiprojections.

It turns out that the only sharp operations of arity $k \geq 4$ are semiprojections.
In other words, given an operation of arity $\geq 4$, if every operation arising from the
identification of two variables is a projection, then these projections
coincide.
\begin{lemma}[\'Swierczkowski's Lemma~\cite{Swierczkowski60:lemma}]
\label{lem:swier}
The only sharp operations of arity $k\geq 4$ are semiprojections.
\end{lemma}

We will need a technical lemma. But first we will introduce some notation. We denote by
$\perm_k$ the set of $k$ cyclic permutations on $\{1,\ldots,k\}$. We denote by $\circ$ the
composition of two permutations, that is, for any $\sigma,\pi\in\perm_k$ we have
$\sigma\circ\pi\in\perm_k$ is defined by $\sigma\circ\pi(x)=\sigma(\pi(x))$. For a $k$-ary operation
$f$ and a permutation $\pi\in\perm_k$ we will denote by $f^\pi$ the
operation $f^\pi=f[\proj^{(k)}_{\pi(1)},\ldots,\proj^{(k)}_{\pi(k)}]$, that is,
$f^\pi(x_1,\ldots,x_k)=f(x_{\pi(1)},\ldots,x_{\pi(k)})$.

\begin{lemma}\label{lem:perm}
Let $W$ be a weighted clone and $\omega\in W$ a positive $k$-ary weighting. Then there
is a positive $k$-ary weighting $\mu\in W$ with the following properties:
\begin{enumerate}
\item $\displaystyle\supp(\mu)=\bigcup_{f\in\supp(\omega)}\bigcup_{\pi\in\perm_k} f^\pi$;
\item $\mu(\proj^{(k)}_i)=-1$ for every $1\leq i\leq k$; 
\item $\mu(f)=\mu(f^\pi)$ for every $f\in\supp(\mu)$ and
$\pi\in\perm_k$.
\end{enumerate}
\end{lemma}
 
\begin{proof}
Let 
\begin{equation}\label{ineq:trick}
\omega'\ =\
\sum_{\pi\in\perm_k}\omega[\proj^{(k)}_{\pi(1)},\ldots,\proj^{(k)}_{\pi(k)}]\,.
\end{equation}
Let $f\in\supp(\omega)$ and
$\pi\in\perm_k$. We have
\begin{multline}\label{eq:long}
\omega'(f)\ =\
\sum_{g\in\supp(\omega)}\
\sum_{\substack{\text{$\sigma\in\perm_k$}\\\text{$g^{\sigma}=f$}}}\omega(g)
\ =\ 
\sum_{g\in\supp(\omega)}\
\sum_{\substack{\text{$\sigma\circ\pi\in\perm_k$}\\\text{$g^{\sigma\circ\pi}=f^\pi$}}}\omega(g)
\ =\ \\
\ =\ 
\sum_{g\in\supp(\omega)}\
\sum_{\substack{\text{$\sigma'\in\perm_k$}\\\text{$g^{\sigma'}=f^\pi$}}}\omega(g)
\ = \ 
\omega'(f^\pi)\,.
\end{multline}
Thus $\omega'$ satisfies the first and the third property of the lemma. 

Since $\omega$ is positive we have that $\sum_{i=1}^k\omega(\proj^{(k)}_i)<0$
and thus, by~\eqref{eq:long}, we have $\omega'(\proj^{(k)}_i)<0$ for every
$1\leq i\leq k$. Let $\omega'(\proj^{(k)}_1)=w$. By~\eqref{eq:long} again,
$\omega'(\proj^{(k)}_i)=w$ for every $1\leq i\leq k$. Thus
$\mu=\frac{1}{w}\omega'$ satisfies all three properties of the lemma.
\end{proof}

\subsection{Cores}

We show that with respect to tractability, the only
interesting weighted clones (and thus weighted relational clones) are those
whose unary weightings can assign positive weight only to very special
operations.

The idea of cores and rigid cores originated in the algebraic approach to
CSPs~\cite{Jeavons98:algebraic,Bulatov05:classifying} and has also proved useful in the complexity
classification of finite-valued CSPs~\cite{hkp14:sicomp,tz13:stoc}.

\begin{definition}
A weighted clone $W$ is a \emph{core} if for every unary weighting $\omega\in W$
every operation $f\in\supp(\omega)$ is bijective.
A valued constraint language $\Gamma$ is a core if $W=\wpol(\Gamma)$ is a
core.
\end{definition}

\begin{theorem}\label{thm:squash}
Let $\Gamma$ be a valued constraint language on $D$. 
If $\Gamma$ is not a core then there is a core valued constraint language 
$\Gamma'$ on $D'\subseteq D$ such that $\Gamma$ is tractable if and only if
$\Gamma'$ is tractable and $\Gamma$ is intractable if and only if $\Gamma'$ is
intractable.
\end{theorem}

\begin{proof}
Let $\omega\in\wpol(\Gamma)$  be a positive unary weighting. By scaling by
$1/|\omega(\proj^{(1)}_1)|$, we have 
$\omega(\proj^{(1)}_1)=-1$ and thus $\sum_{f\in\supp(\omega)}\omega(f)=1$. For
any weighted relation $\phi\in\Gamma$ of arity $m$ and any $m$-tuple
$\tup{x}\in
D^m$, we have $(\ast)$
$\phi(\tup{x})\geq\sum_{f\in\supp(\omega)}\omega(f)\phi(f(\tup{x}))$. Suppose
that $\tup{y}$ is a
minimal-cost assignment for $\phi$; that is, $\phi(\tup{y})\leq\phi(\tup{x})$ for all
$\tup{x}\in
D^m$. Then for every $f\in\supp(\omega)$, we have $f(\tup{y})$ is also a minimal-cost
assignment. Assume for contrary that for some $f'\in\supp(\omega)$, we have
$\phi(f'(\tup{y}))>\phi(\tup{y})$; write $\phi(f'(\tup{y}))=\phi(\tup{y})+\epsilon$, where $\epsilon>0$.
Then we claim that there is an $f\in\supp(\omega)$ such that
$\phi(f(\tup{y}))<\phi(\tup{y})$, which contradicts the choice of $\tup{y}$. To prove the claim
assume that $\phi(f(\tup{y}))\geq\phi(\tup{y})$ for every
$f\in\supp(\omega)\setminus\{f'\}$. Hence we get
$\sum_{f\in\supp(\omega)}\omega(f)\phi(f(\tup{y})) =
\sum_{f\in\supp(\omega)\setminus\{f'\}}\omega(f)\phi(f(\tup{y}))+\omega(f')\phi(f'(\tup{y}))
\geq (1-\omega(f')\phi(\tup{y})+\omega(f')(\phi(\tup{y})+\epsilon) =
\phi(\tup{y})+\omega(f')\epsilon>\phi(\tup{y})$, which contradicts $(\ast)$.

Consequently, given an instance $I\in\VCSP(\Gamma)$ and a solution $s$ to
$I$, we can take any unary weighting $\omega\in\wpol(\Gamma)$ and any unary operation
$f\in\supp(\omega)$ and get another solution $f(s)$ to $I$; the solution
$f(s)$ uses only labels from $f(D)$.
Consider a unary non-bijective operation $f\in\supp(\omega)$ with the minimum
$|f(D)|$ over all unary weightings $\omega\in\wpol(\Gamma)$.
We denote by $D'=f(D)$ the range of $f$. We denote by $\Gamma'$ language
containing the restriction of every $\phi\in\Gamma$ to $D'$.

Given any instance $I\in\VCSP(\Gamma)$ we can create, by replacing
each weighted relation $\phi$ in $I$ by $\phi'$, in polynomial time an instance
$I'\in\VCSP(\Gamma')$ with the following properties: any solution to $I'$ is
also a solution to $I$, and for any solution $s$ to $I$ we have that $f(s)$ is a
solution to $I'$. If $\Gamma'$ is not a core we can repeat the same construction
with $\Gamma'$.
\end{proof}

Theorem~\ref{thm:squash} was independently obtained
in~\cite{Kozik15:icalp}, where it was also shown that, with respect to
tractability, it suffices to restrict to rigid cores.

\begin{definition}
A weighted clone $W$ is a \emph{rigid core} if the only unary operation in the
support clone of $W$ is the unary projection $\proj^{(1)}_1$.
A valued constraint language $\Gamma$ is a rigid core if $W=\wpol(\Gamma)$ is a
rigid core; that is if the only unary polymorphism of $\Gamma$ is $\proj^{(1)}_1$.
\end{definition}

\begin{theorem}[\hspace*{-0.3em}\cite{Kozik15:icalp}]\label{thm:rigid}
Let $\Gamma$ be a valued constraint language on $D$. 
If $\Gamma$ is not a rigid core then there is a rigid core valued constraint language 
$\Gamma'$ on $D'\subseteq D$ such that $\Gamma$ is tractable if and only if
$\Gamma'$ is tractable and $\Gamma$ is intractable and only if $\Gamma'$ is
intractable.
\end{theorem}

It is not hard to show that a weighted clone $W$ (a valued constraint language
$\Gamma$) is a rigid core if and only if all operations in the support clone of $W$
(polymorphisms of $\Gamma$, respectively) are idempotent.

\section{Conditions for Tractability}
\label{sec:wclones}

In this section we will present our main results.

Creed and \v{Z}ivn\'y obtained the following result on the structure of 
weighted clones with a positive weighting~\cite[Theorem~2]{cz11:cp-mwc}; see also
\cite[Corollary~7.7]{cccjz13:sicomp}.

\begin{theorem}[\hspace*{-0.3em}\cite{cz11:cp-mwc}]~\label{thm:class1}
Any weighted clone $W$ containing a positive weighting
contains a weighting whose support is either:
\begin{enumerate}
  \item a set of unary operations that are not projections; or
  \item a set of binary idempotent operations that are not projections; or
  \item a set of ternary operations that are majority operations, minority operations, Pixley operations or semiprojections; or
  \item a set of $k$-ary semiprojections (for some $k > 3$).
\end{enumerate}

\end{theorem}

Since rigid cores require all unary weightings be zero-valued, case~(1) of
Theorem~\ref{thm:class1} can be easily eliminated. Moreover, using Gordan's
Theorem (a variant of Farkas' Lemma) we will strengthen Theorem~\ref{thm:class1}
by refining the ternary case, thus obtaining the following result, which is the
main result of this paper.

\begin{theorem}[\textbf{Main}]\label{thm:class2}
Any weighted clone $W$ that is a rigid core and contains a positive weighting
also contains a weighting whose support is either:
\begin{enumerate}
  \item a set of binary idempotent operations that are not projections; or
  \item a set of ternary operations that are either:
  \begin{enumerate}
   \item a set of majority operations; or
   \item a set of minority operations; or
  \item  a set of majority operations with total weight $2$ and a set
  of minority operations with total weight $1$; or
  \end{enumerate}
  \item a set of $k$-ary semiprojections (for some $k \geq 3$).
\end{enumerate}
\end{theorem}

The proof of Theorem~\ref{thm:class2} can be found in Section~\ref{sec:proof}.

Note that compared to Theorem~\ref{thm:class1} the inequality in case~(3) of 
Theorem~\ref{thm:class2} is not strict as it includes one of the ternary cases.

We remark that Theorem~\ref{thm:class2} holds for \emph{any} weighted clone $W$
with \emph{any} support clone $C$ as long as $W$ contains a positive weighting.

Theorem~\ref{thm:class2} tells us that (i) Pixley operations are not necessary
for tractability, (ii) semiprojections can be separated from the other types of
ternary operations,  and (iii) the only possible interplay between majority and
minority operations, as described in case~(2c) of Theorem~\ref{thm:class2}. 

We now focus on the weighted clones containing one of the weightings described
in Theorem~\ref{thm:class2}.

\paragraph{Case~(1) of Theorem~\ref{thm:class2}}

A weighting described in Theorem~\ref{thm:class2}\,(1) can lead both to
tractable and intractable weighted clones, as the next two examples demonstrate,
but the precise boundary of tractability is currently unknown.

\begin{example}\label{ex:TP}
A binary operation $f:D^2\rightarrow D$ is called conservative if $f(x,y)\in\{x,y\}$ for all $x,y\in D$ and commutative if $f(x,y)=f(y,x)$ for
all $x,y\in D$. Moreover, $f$ is called a \emph{tournament} operation if $f$ is both conservative
and commutative. Let $W$ be a weighted clone such $\supp(W)$ contains a
tournament operation. Then, by a recent result of the
authors~\cite{tz15:icalp}, $W$ is tractable.
\end{example}

\begin{example}\label{ex:f}
Let $D=\{0,1,2,3\}$ and let $f$ be a binary operation defined by 
Table~\ref{table:bin}.
\begin{table}[hbt]
\begin{center}
\begin{tabular}{c|cccc}
\textbf{f} & \textbf{0} & \textbf{1} & \textbf{2} & \textbf{3} \\ \hline
\textbf{0} & 0 & 1 & 0 & 1 \\
\textbf{1} & 0 & 1 & 0 & 1 \\
\textbf{2} & 2 & 3 & 2 & 3 \\
\textbf{3} & 2 & 3 & 2 & 3 
\end{tabular}
\end{center}
\caption{Definition of $f$ from Example~\ref{ex:f}.}
\label{table:bin}
\end{table}
Note that $f$ is an idempotent operation but not a projection. In fact, $f$ is
an example of a \emph{rectangular band}~\cite{McKenzie87:algebras},
which is an idempotent and associative binary operation $f:D^2\to D$ satisfying
$f(x,f(y,z))=f(x,z)$ for all $x,y,z\in D$. Let $W$ be the weighted clone
generated by the weighting $\omega$ defined by
$\omega(\proj^{(2)}_1)=\omega(\proj^{(2)}_2)=-1$ and $\omega(f)=+2$. It is known
that $W$ is intractable~\cite{Jeavons97:closure,Pearson97:survey}.
\end{example}

\paragraph{Case~(2a) of Theorem~\ref{thm:class2}}

A weighting as described in Theorem~\ref{thm:class2}\,(2a) implies tractable
weighted clones, as we will now show.

A weighted relational clone that contains only relations (and thus is a
relational clone) is called \emph{crisp}. A weighted clone $W$ is called crisp
if $\imp(W)$ is a crisp weighted relational clone.

\begin{proposition}\label{prop:majority}

Let $W$ be a weighted clone with a positive ternary weighting $\omega\in W$ such that all
operations $f\in\supp(\omega)$ are majority operations. Then $W$ is
crisp.

\end{proposition}

In order to prove Proposition~\ref{prop:majority}, we prove a more general
result. A $k$-ary operation $f:D^k\to D$, where $k\geq 3$, is called a
\emph{near-unanimity} operation if for all $x,y\in D$,
\begin{equation}\label{eq:defNU}
f(y,x,x,\ldots,x)=f(x,y,x,x,\ldots,x)=\cdots=f(x,x,\ldots,x,y)=x\,.
\end{equation}

Note that a ternary near-unanimity operation is a majority operation.

\begin{proposition}\label{prop:NU}
Let $W$ be a weighted clone with a positive weighting $\omega\in W$ such that all
operations $f\in\supp(\omega)$ are near-unanimity operations. Then $W$ is crisp.
\end{proposition}
\begin{proof}
Let $\omega$ be $k$-ary. Note that if $f$ is a $k$-ary near-unanimity operation
then so is $g(x_1,\ldots,x_k)=f(x_{\pi(1)},\ldots,x_{\pi(k)})$ for any
permutation $\pi$ on $\{1,\ldots,k\}$. Thus, by Lemma~\ref{lem:perm}, we can
assume $\omega$ assigns weight $-1$ to each of the $k$ projections (and still every
$f\in\supp(\omega)$ is a near-unanimity operation).

Let $\phi\in\imp(W)$ be an $m$-ary weighted relation and let
$\tup{x},\tup{y}\in D^m$ be such that $\tup{x},\tup{y}\in\dom(\phi)$. 
Since $\omega\in\wpol(\phi)$, we have, by~\eqref{eq:WPOL} with
$\tup{x}_1=\tup{y}$ and $\tup{x}_i=\tup{x}$ for all $2\leq i\leq k$, and by~\eqref{eq:defNU},
$-\phi(\tup{y})-(k-1)\phi(\tup{x})+k\phi(\tup{x})\leq 0$, which gives $\phi(\tup{x})\leq\phi(\tup{y})$. By swapping
$\tup{x}$ and $\tup{y}$ in~\eqref{eq:WPOL}, we get
$-\phi(\tup{x})-(k-1)\phi(\tup{y})+k\phi(\tup{y})\leq 0$, which gives
$\phi(\tup{y})\leq\phi(\tup{x})$. Together, $\phi(\tup{x})=\phi(\tup{y})$ for all
$\tup{x},\tup{y}\in\dom(\phi)$. 
\end{proof}

Since crisp weighted relational clones with a near-unanimity polymorphism 
are tractable~\cite{Jeavons97:closure}, we get
the following.

\begin{corollary}
A weighted clone containing a positive weighting $\omega$ with all operations in
$\supp(\omega)$ being near-unanimity operations is tractable.
\end{corollary}

\paragraph{Case~(2b) of Theorem~\ref{thm:class2}}

A weighting as described in Theorem~\ref{thm:class2}\,(2b) also implies tractable
weighted clones, as we will now show.

\begin{proposition}\label{prop:minority}

Let $W$ be a weighted clone with a positive ternary weighting $\omega\in W$ such that all
operations $f\in\supp(\omega)$ are minority operations. Then $W$ is crisp.

\end{proposition}

\begin{proof}
Note that if $f$ is a minority operation
then so is $g(x_1,x_2,x_3)=f(x_{\pi(1)},x_{\pi(2)},x_{\pi(3)})$ for any
permutation $\pi$ on $\{1,2,3\}$. Thus, by Lemma~\ref{lem:perm}, we can
assume $\omega$ assigns weight $-1$ to each of the three projections (and still every
$f\in\supp(\omega)$ is a minority operation).

Let $\phi\in\imp(W)$ be an $m$-ary weighted relation and let
$\tup{x},\tup{y}\in D^m$ be such that $\tup{x},\tup{y}\in\dom(\phi)$. 
Since $\omega\in\wpol(\phi)$, we have, by~\eqref{eq:WPOL} with
$\tup{x}_1=\tup{x}$ and $\tup{x}_2=\tup{x}_3=\tup{y}$,
$-\phi(\tup{x})-2\phi(\tup{y})+3\phi(\tup{x})\leq 0$, which gives $\phi(\tup{x})\leq\phi(\tup{y})$. By swapping
$\tup{x}$ and $\tup{y}$ in~\eqref{eq:WPOL}, we get
$-\phi(\tup{y})-2\phi(\tup{x})+3\phi(\tup{y})\leq 0$, which gives
$\phi(\tup{y})\leq\phi(\tup{x})$. Together, $\phi(\tup{x})=\phi(\tup{y})$ for all
$\tup{x},\tup{y}\in\dom(\phi)$. 
\end{proof}

Since crisp weighted relational clones with a minority polymorphism are
tractable~\cite{Jeavons97:closure}, we get the following.

\begin{corollary}
A weighted clone containing a positive weighting $\omega$ with all operations in
$\supp(\omega)$ being minority operations is tractable.
\end{corollary}

\paragraph{Case~(2c) of Theorem~\ref{thm:class2}}

In a recent paper the authors have shown~\cite{tz15:icalp} that any weighting
described in Theorem~\ref{thm:class2}\,(2c) implies tractability. This is a
corollary of the following result.

\begin{theorem}[\hspace*{-0.3em}\cite{tz15:icalp}]
Let $W$ be a weighted clone. If there is a weighting $\omega\in W$ such that
$\supp(\omega)$ contains a majority operation then $W$ is tractable.
\end{theorem}

Previously, only a special type of the weightings described in
Theorem~\ref{thm:class2}\,(2c) has been known to imply tractability.

\begin{example}
A $k$-ary weighting $\omega$ is a \emph{multimorphism} if
$\omega(f)\in\mathbb{N}$ for all $f\in\supp(\omega)$ and
$\omega(\proj^{(k)}_i)=-1$ for all $1\leq i\leq
k$~\cite{Cohen06:complexitysoft}. It has been shown that if a weighted clone
$W$ contains a weighting $\omega$ described in Theorem~\ref{thm:class2}\,(2c)
such that $\omega$ is a multimorphism then $W$ is tractable~\cite{kz13:jacm}.
\end{example}

\paragraph{Case~(3) of Theorem~\ref{thm:class2}}

We show that the weightings described in Theorem~\ref{thm:class2}\,(3)
\emph{alone} are not
sufficient for tractability. As in case~(1),
the precise boundary of tractability is currently unknown.

\begin{example}
Let $D$ be a fixed set with $|D|>2$. 
Fix two distinct labels from $D$, say $0,1\in D$.
Let $\phi$ be the following ternary weighted relation: $\phi(x,y,z)=\infty$ if
$\{x,y,z\}=\{0\}$, or $\{x,y,z\}=\{1\}$, or $\{x,y,z\}\neq\{0,1\}$; and
$\phi(x,y,z)=0$ otherwise. The weighted relation $\phi$ corresponds to the
\textsc{Not-All-Equal Satisfiability} problem~\cite{Garey79:intractability},
which is NP-hard~\cite{Schaefer78:complexity}. It is easy to show that every
semiprojection on $D$ is a polymorphism of $\phi$.
Take a $k$-ary semiprojection $f$ for some $k\geq 3$ and
$\tup{x}_1,\ldots,\tup{x}_k\in\feas(\phi)$. From the definition of $\phi$, 
we have $\tup{x}_i\in\{0,1\}^3$ for every
$1\leq i\leq k$. Since there are at most two distinct labels in each 
coordinate, $f(\tup{x}_1,\ldots,\tup{x}_k)$ reduces to a projection (from the
definition of semiprojections) and thus $f$ is a polymorphism of $\phi$ as
$f(\tup{x}_1,\ldots,\tup{x}_k)=\tup{x}_i$ for some $1\leq i\leq k$.

Let $C$ be the clone of operations on $D$ generated by all semiprojections on
$D$. Let $W=\mathbf{W}_C$ be the weighted
clone containing all possible weightings of $C$. In particular, $W$ contains all
possible weightings whose support contains only semiprojections. Since
$C\subseteq\pol(\phi)$ and $\phi$ is a relation we have that $\phi\in\imp(W)$.
Consequently, $W$ is intractable.
\end{example}

\paragraph{Finite-Valued Weighted Clones}

Recall that valued constraint languages capture both decision and optimisation
problems. Clones, which capture crisp valued constraint languages and thus
purely decision problems, have been studied extensively in universal
algebra~\cite{Szendrei86:clones,Hobby88:structure}. We now focus on an important
special type of weighted clones that correspond to valued constraint languages
that capture purely optimisation problems. Such valued constraint languages are
called \emph{finite-valued} as they only contain finite-valued weighted
relations.

Weighted clones corresponding to finite-valued constraint languages (together
with the binary equality relation $\phi_=$) are those 
with support clone $\mathcal{O}_D$. To see this,
we denote, for a clone $C$, by $\inv(C)$ the relational clone that consists of relations $R$ with $f\in\pol(R)$ for every $f\in C$.
Then, it is well known that
$\inv(\mathcal{O}_D)=\relclone(\{\phi_=\})$ and observe that
$\feas(\imp(W))\subseteq\inv(\mathcal{O}_D)$ for any weighted clone $W$ with
support clone $\mathcal{O}_D$.\footnote{More generally, we have
$\feas(\imp(W))=\inv(C)$ for any nonempty weighted clone $W$ with support clone
$C$. On the one hand, if $\phi\in\imp(W)$ then, by Definition~\ref{def:wp},
$C\subseteq\pol(\phi)$, which implies
$\feas(\phi)\in\inv(\pol(\phi))\subseteq\inv(C)$. On the other hand, if
$R\in\inv(C)$ then $C\subseteq\pol(R)$. Since $R$ satisfies~\eqref{eq:WPOL} we
have $R\in\imp(W)$.}

However, as we have limited our scope to rigid cores (which, by
Theorem~\ref{thm:rigid}, does not change tractability), we will define a
weighted clone $W$ to be finite-valued if its support clone is equal to
$\mathbf{I}_D$, the clone of all \emph{idempotent} operations on $D$. 

\begin{definition} 
A weighted clone $W$ on $D$  is called \emph{finite-valued} if the support clone of $W$ is $\mathbf{I}_D$. 
\end{definition}

For any $d\in D$, the unary constant relation $\phi_d$ is defined by
$\phi_d(x)=0$ if $x=d$ and $\phi_d(x)=\infty$ otherwise. Let
$\mathcal{R}=\relclone(\{\phi_=\}\cup\cup_{d\in D}\{\phi_d\})$. It is known that
$\inv(\mathbf{I}_D)=\mathcal{R}$~\cite{Bulatov05:classifying}.

The weighted relational clones corresponding to finite-valued weighted clones
are those that are subsets of the weighted relational clone generated by
$\mathcal{R}$ and finite-valued weighted relations.

We already know that weighted clones containing any of the weightings described
in Theorem~\ref{thm:class2}\,(2a-c) are tractable. In fact, in the finite-valued
case, the corresponding weighted relational clones are rather trivial as we will
now show.

Let $W$ be a finite-valued weighted clone on $D$. Then for any weighted
relation $\phi\in\imp(W)$ we have $\feas(\phi)\in\mathcal{R}$.

If $W$ contains a weighting described in Theorem~\ref{thm:class2}\,(2a) then, by
Proposition~\ref{prop:majority}, $\imp(W)$ is crisp and thus every
$\phi\in\imp(W)$ can be written as the addition of a rational constant to a
weighted relation in $\mathcal{R}$.
Hence $\imp(W)$ is tractable. Similarly, if $W$ contains a weighting
described in Theorem~\ref{thm:class2}\,(2b) then, by
Proposition~\ref{prop:minority}, $\imp(W)$ is crisp and thus every
$\phi\in\imp(W)$ can be written as the addition of a rational constant to a
weighted relation in $\mathcal{R}$.
Hence $\imp(W)$ is tractable.

The next result shows that a weighting described in
Theorem~\ref{thm:class2}\,(2c) also suffices for tractability in the
finite-valued case.

\begin{proposition}
Let $W$ be a finite-valued weighted clone. If $W$ contains a positive weighting
described in Theorem~\ref{thm:class2}\,(2c) then every weighted relation 
$\phi\in\imp(W)$ can be expressed as a sum of unary weighted relations and the
binary equality relation $\phi_=$.
\end{proposition}

\begin{proof}
By Lemma~\ref{lem:perm}, we can assume the weighting assigns weight $-1$ to each
of the three projections and still is as described in
Theorem~\ref{thm:class2}\,(2c).

An $m$-ary relation $R$ on $D$ is called trivial if $R=D^m$. First we show than
any relation $R\in\mathcal{R}$ can be expressed as a sum of unary relations,
trivial relations, and $\phi_=$. The claim holds true for the generators of
$\mathcal{R}$, that is, for $\phi_=$ and $\phi_d$ for all $d\in D$. Next, if
$R=R_1\wedge R_2$ and the claim holds true for both $R_1$ and $R_2$ the it also
holds true for $R$. Finally, let $R=\exists x R'$ and assume that $R'$ satisfies
the claim. If $x$ appears in some $\phi_=$ in $R'$, say $\phi_=(x,x')$, then we
can replace all occurrences of $x$ by $x'$. Otherwise, $x$ appears only in
constant and trivial relations in $R'$. If the conjunction of the unary
relations that $x$ appears in is empty then the claim holds trivially.
Otherwise, we can replace $x$ by any other variable.

Consequently, for any $\phi\in\imp(W)$, $\feas(\phi)$ can be written as a sum of
unary relations, trivial relations, and $\phi_=$. Observe that any $\phi$ with
$\feas(\phi)=\phi_=$ can be written as a sum of $\phi_=$ and the unary weighted
relation $\phi'(x)=\phi(x,x)$. Thus, it remains to show that any
$\phi\in\imp(W)$ with $\feas(\phi)$ being a trivial relation can be written as a
sum of unary weighted relations.

For any $m$-tuple $\tup{x}\in D^m$, we will write $\tup{x}[i\leftarrow d]$ to
denote the tuple with $d\in D$ {substituted} at position $i$. In other words,
$\tup{x}[i\leftarrow d]$ is the $m$-tuple identical to $\tup{x}$ except
(possibly) at position $i$, where it is equal to $d$.

We will use~\cite[Lemma~6.23]{Cohen06:complexitysoft} which says that a weighted
relation $\phi:D^m\to\Q$ can be expressed as a sum of unary weighted relations if and
only if, for all $\tup{x},\tup{y}\in D^m$ and all $1\leq i\leq m$, we have
\begin{equation}
\phi(\tup{x})+\phi(\tup{y}) = \phi(\tup{x}[i\leftarrow \tup{y}[i]]) +
\phi(\tup{y}[i\leftarrow \tup{x}[i]])\,.
\end{equation}

Take any $\tup{x},\tup{y}\in D^m$ and $1\leq i\leq m$. 
Let $a=\tup{x}[i]$ and $b=\tup{y}[i]$. Now consider the tuples $\tup{x}$,
$\tup{x}[i\leftarrow b]$, and $\tup{y}[i\leftarrow a]$. 
By applying the
weighting from the statement of the proposition as in~\eqref{eq:WPOL}, we get
$\phi(\tup{x})+\phi(\tup{x}[i\leftarrow b])+\phi(\tup{y}[i\leftarrow a]) \geq
2\phi(\tup{x})+\phi(\tup{y})$ and thus $\phi(\tup{x}[i\leftarrow
b])+\phi(\tup{y}[i\leftarrow a]) \geq \phi(\tup{x})+\phi(\tup{y})$.
Now consider the tuples $\tup{x}$, $\tup{y}$, and $\tup{y}[i\leftarrow a]$. By
applying the weighting from the statement of the proposition as
in~\eqref{eq:WPOL}, we get
$\phi(\tup{x})+\phi(\tup{y})+\phi(\tup{y}[i\leftarrow a]) \geq 2\phi(\tup{y}[i\leftarrow a]) + \phi(\tup{x}[i\leftarrow b])$ and thus
$\phi(\tup{x})+\phi(\tup{y}) \geq \phi(\tup{y}[i\leftarrow a]) + \phi(\tup{x}[i\leftarrow b])$.
\end{proof}

\begin{corollary}
A finite-valued weighted clone containing a positive weighting $\omega$ described
in~Theorem~\ref{thm:class2}\,(2c) is tractable.
\end{corollary}

The only remaining finite-valued weighted clones contain a weighting that is
either as described in Theorem~\ref{thm:class2}\,(1) or as described in
Theorem~\ref{thm:class2}\,(3). We have seen an example of a (tractable) weighted
clone with a weighting as described in Theorem~\ref{thm:class2}\,(1) in
Example~\ref{ex:TP}. 

We now give an example of an \emph{intractable finite-valued} weighted clone
with a weighting as described in Theorem~\ref{thm:class2}\,(1). (We note that
the intractability of the weighted clone $W$ from Example~\ref{ex:f}, which
contains a weighting as described in Theorem~\ref{thm:class2}\,(1), relies on
the fact that $W$ is not finite-valued and thus is not immediately applicable
here.)

\begin{example}
Let $D=\{0,1,2\}$.
Recall from Example~\ref{ex:TP} that a binary operation $f:D^2\to D$ is \emph{conservative} if $f(x,y)\in\{x,y\}$ for
all $x,y\in D$. 
For any conservative binary operation $f:D^2\to D$ and any 2-element subdomain
$\{a,b\}\subseteq D$, the restriction $\Crestrict{f}{\{a,b\}}$ of $f$ onto $\{a,b\}$
behaves either as $\proj^{(2)}_1$,
$\proj^{(2)}_2$, $\min$, or $\max$, where $\min$ and $\max$ are the two operations
that return the smaller (larger) of its two arguments with respect to the usual order $0<1<2$, respectively.
Consider the operations in Table~\ref{table:w} described by their behaviour on
the various 2-element subdomains.

\begin{table}[htb]
\begin{center}
\begin{tabular}{lccc|l}
$\mathbf{f}$ & \{\textbf{0},\textbf{1}\} & \{\textbf{0},\textbf{2}\} & \{\textbf{1},\textbf{2}\} & $\mathbf{\omega(f)}$\\ \hline
$f_1$ & $\proj^{(2)}_1$ & $\proj^{(2)}_1$ & $\proj^{(2)}_1$ & -0.5 \\
$f_2$ & $\proj^{(2)}_2$ & $\proj^{(2)}_2$ & $\proj^{(2)}_2$ & -0.5 \\
$f_3$ & $\proj^{(2)}_1$ & $\min$ & $\proj^{(2)}_1$ & 0.5 \\
$f_4$ & $\proj^{(2)}_2$ & $\min$ & $\proj^{(2)}_2$ & 0.5 \\
\end{tabular}
\end{center}
\caption{Definition of $\omega$}\label{table:w}
\end{table}

Note that $f_1=\proj^{(2)}_1$ and $f_{2}=\proj^{(2)}_2$.
The weighting $\omega$ is defined by the last column of
Table~\ref{table:w}.
Note that $\omega$ is not commutative.
It can be checked that $\omega$ is a weighted polymorphism
of the finite-valued weighted relation $\phi : \{0,1,2\} \to \Q$ defined in Table~\ref{table:phi}.

\begin{table}[hbt]
\begin{center}
\begin{tabular}{c|ccc}
$\phi$ & \textbf{0} & \textbf{1} & \textbf{2} \\ \hline
\textbf{0} & 1 & 0 & 1 \\
\textbf{1} & 0 & 1 & 1 \\
\textbf{2} & 1 & 1 & 1 
\end{tabular}
\end{center}
\caption{Definition of $\phi$}\label{table:phi}
\end{table}

Now since $\argmin \phi=\{(0,1),(1,0)\}$, we have that $\phi$ satisfies the (MC)
condition~\cite{tz13:stoc} and thus can be used to reduce from
Max-Cut~\cite{Cohen06:complexitysoft}. Thus, $W=\wclone(\{\omega\})$ is
intractable.
\end{example}

Thus weightings described in Theorem~\ref{thm:class2}\,(1) can lead to both
tractable and intractable finite-valued weighted clones. 
The authors have recently shown that with respect to tractability of
finite-valued constraint languages the necessary and sufficient condition is
having a binary weighting $\omega$ that assigns positive weight to idempotent
\emph{commutative} operations only; that is, for every $f\in\supp(\omega)$ we
have $f(x,y)=f(y,x)$ for all $x,y\in D$~\cite{tz13:stoc}.\footnote{The result
from~\cite{tz13:stoc} extends from finite-valued constraint languages to
finite-valued weighted relational clones as adding the binary equality relation
and unary constant relations does not affect tractability in the presence of a
binary commutative weighted polymorphism.}
However, the precise interplay of case~(1) and case~(3) of
Theorem~\ref{thm:class2} is currently unknown.

\paragraph{Minimal Weighted Clones}

Any weighted relational clone $\Gamma\subseteq\wrel_D$ with
$\wrelclone(\Gamma)=\wrel_D$ is NP-hard. 
A weighted relational clone on $D$ is called \emph{maximal} if it is as large as
possible but $\wrelclone(\Gamma)\neq\wrel_D$.\footnote{A (tractable) valued
constraint language $\Gamma$ is called maximal in~\cite{Cohen06:complexitysoft}
if for any $\phi\not\in\Gamma$, $\Gamma\cup\{\phi\}$ is intractable. We require
$\wrelclone(\Gamma\cup\{\phi\})=\wrel_D$, which implies the intractability of
$\Gamma\cup\{\phi\}$, thus borrowing the concept of maximality
from~\cite{Bulatov01:complexity,Bulatov04:graph,Jonsson08:siam,Bodirsky09:tcs}
and extending it from relational clones to weighted relational clones.} 

\begin{definition}\label{def:maximal}
A weighted relational clone $\Gamma\subseteq\wrel_D$ is called \emph{maximal}
if $\wrelclone(\Gamma)\neq\wrel_D$ but for any $\phi\not\in\Gamma$ we have
$\wrelclone(\Gamma\cup\{\phi\})=\wrel_D$. 
\end{definition}

It follows that a valued constraint language $\Gamma$ is maximal if and only if
the weighted relational clone $\wrelclone(\Gamma)$ is maximal.

As a special case of Definition~\ref{def:maximal}, we get that a relational clone
$\Gamma$ is maximal if $\Gamma\neq\rel_D$ but for any $R\in\rel_D$ we have
$\relclone(\Gamma\cup\{R\})=\rel_D$.

A weighted clone is called \emph{minimal} if it is not zero-valued but the only
weighted clone properly included in it is the zero-valued weighted clone. 

\begin{definition}
A weighted clone $W$ with support clone $C$ is called \emph{minimal} if
$W\neq\wops_C^0$ and every positive weighting $\omega\in W$ satisfies
$\wclone(\omega)=W$.
\end{definition}

Maximal weighted relational clones correspond, via the Galois correspondence
given in Theorem~\ref{thm:wgc}, to minimal weighted clones. 

We will be interested in maximal \emph{tractable} weighted relational clones and
thus minimal tractable weighted clones.
Maximal crisp weighted relational clones have been classified with respect to
tractability in~\cite{Bulatov01:complexity,Bulatov04:graph}.
We now show that that there are no tractable maximal non-crisp weighted
relational clones.

\begin{theorem}
All maximal non-crisp weighted relational clones are intractable.
\end{theorem}

\begin{proof}
If $\Gamma$ contains all finite-valued weighted relations then it is
intractable. Otherwise, there is a finite-valued weighted relation
$\phi\not\in\Gamma$. 
Since $\phi$ is finite-valued we have $\feas(\Gamma)=\feas(\wrelclone(\Gamma\cup\{\phi\}))$. But then either
$\feas(\Gamma)=\mathbf{R}_D$, in which case $\Gamma$ is intractable, or
$\feas(\Gamma)=\feas(\wrelclone(\Gamma\cup\{\phi\}))\neq\mathbf{R}_D$, in which
case $\Gamma$ is not maximal.
\end{proof}

\section{Proof of Theorem~\ref{thm:class2}}\label{sec:proof}

In this section we will prove the following theorem, which is our main result.

\begin{theorem*}[Theorem~\ref{thm:class2} restated]
Any weighted clone $W$ that is a rigid core and contains a positive weighting
also contains a weighting whose support is either:
\begin{enumerate}
  \item a set of binary idempotent operations that are not projections; or
  \item a set of ternary operations that are either:
  \begin{enumerate}
   \item a set of majority operations; or 
   \item a set of minority operations; or 
  \item  a set of majority operations with total weight $2$ and a set
  of minority operations with total weight $1$; or 
  \end{enumerate}
  \item a set of $k$-ary semiprojections (for some $k \geq 3$).
\end{enumerate}
\end{theorem*}

We will use the following variant of Farkas' Lemma.
\begin{theorem}[Gordan]\label{thm:gordan}
Let $A\in\mathbb{Q}^{n\times m}$ be a matrix. 
Either $Ax=0$, where $x\in\mathbb{Q}^{m}$ with $x\geq 0$
and $x\neq 0$, or $\exists y\in\mathbb{Q}^n$ with $y^{\intercal}A>0$.
\end{theorem}

By Definition~\ref{defn:wclone}, only \emph{proper} superpositions are allowed
within a weighted clone. However, the following result
from~\cite{cccjz13:sicomp} shows that any weighted sum of arbitrary
superpositions of a pair of weightings $\omega_1$ and $\omega_2$ can be obtained
by taking a weighted sum of superpositions of $\omega_1$ and $\omega_2$ with
projection operations, and then taking a superposition of the result. Given that
superpositions with projections are always proper~\cite{cccjz13:sicomp}, this
result implies that any weighting which can be expressed as a weighted sum of
arbitrary (i.e., possibly improper) superpositions can also be expressed as a superposition of a weighted
sum of \emph{proper} superpositions.

\begin{lemma}[\protect{\hspace*{-0.1cm}\cite[Lemma~6.4]{cccjz13:sicomp}}]
\label{lemma:sumofsuperpos}
Let $C$ be a clone, and let $\omega_1$ and $\omega_2$ be weightings of $C$, of arity $k$ and $\ell$ respectively.
For any $g_1,\ldots,g_k \in C^{(m)}$ and any $g'_1,\ldots,g'_\ell \in C^{(m)}$,
\begin{equation*}
c_1 \, \omega_1[g_1,\ldots,g_k] + \ c_2 \, \omega_2[g'_1,\ldots,g'_{\ell}] =
\omega[g_1,\ldots,g_k,g'_1,\ldots,g'_{\ell}]\,,
\end{equation*}
where $\omega = c_1 \, \omega_1[e^{(k+\ell)}_1,\ldots,e^{(k+\ell)}_k] + \ c_2 \, \omega_2[e^{(k+\ell)}_{k+1},\ldots,e^{(k+\ell)}_{k+\ell}]$
\end{lemma}

Before proving Theorem~\ref{thm:class2} we introduce the following useful notion.
For the reader's convenience, we repeat here Table~\ref{tab:sharp3} from
Section~\ref{sec:prelim}.
\newcounter{savetable}
\setcounter{savetable}{\value{table}}
  \begin{table}[ht]
  \begin{center}
  \begin{tabular}{c|*{8}{c}}
    \makebox[1.5cm]{Input} 
    & \makebox[0.75cm]{Mj} & \makebox[0.75cm]{S1} & \makebox[0.75cm]{S2}
    & \makebox[0.75cm]{S3} & \makebox[0.75cm]{P1} & \makebox[0.75cm]{P2}
    & \makebox[0.75cm]{P3} & \makebox[0.75cm]{Mn}\\
    \hline
    (x,x,y) & x & x & x & y & x & y & y & y\\
    (x,y,x) & x & x & y & x & y & x & y & y\\
    (y,x,x) & x & y & x & x & y & y & x & y\\
  \end{tabular}
  \end{center}
  \caption*{Table 1 (restated): Sharp ternary operations}
  \end{table}
\setcounter{table}{\value{savetable}}
We call (ternary) operations corresponding to columns 5, 6, and 7 in
Table~\ref{tab:sharp3} Pixley operations of type 1, 2, and 3 respectively, and
will denote by P1 (P2 and P3) the Pixley operations of type 1 (2 and 3,
respectively). We call (ternary) semiprojections corresponding to columns 2, 3, and 4 in
Table~\ref{tab:sharp3} semiprojections of type 1, 2, and 3 respectively, and
will denote by S1 (S2 and S3) the semiprojections of type 1 (2 and 3,
respectively). More generally, a $k$-ary semiprojection $f$ is
called of type $1\leq i\leq k$ if equating any two inputs of $f$ results in
$\proj^{(k)}_i$.

For any Pixley operation $f$ of type $i\in\{1,2,3\}$ we can obtain, by
(cyclically) permuting the arguments of $f$, Pixley operations of the other two
types. For instance, if $f\in\mbox{P1}$ then we have $g\in\mbox{P2}$ and
$h\in\mbox{P3}$, where
$g(x,y,z)=f[\proj^{(3)}_3,\proj^{(3)}_1,\proj^{(3)}_2]=f(z,x,y)$ and
$h(x,y,z)=f[\proj^{(3)}_2,\proj^{(3)}_3,\proj^{(3)}_1]=f(y,z,x)$.
Two Pixley operations $f$ and $g$ of different types are called \emph{related}
if there is a permutation $\pi\in\perm_3$ such that $f=g^\pi$. (Note that the
requirement of $f$ and $g$ being of different types excludes the identity
permutation $(1,2,3)$ and there are only other two permutations in $\perm_3$,
namely $(2,3,1)$ and $(3,1,2)$.)

Similarly, two semiprojections $f$ and $g$ of different types are called related
if there is a permutation $\pi\in\perm_3$ such that $f=g^\pi$.

The following table, which can be verified using the definitions above, will be
useful in the proof of Theorem~\ref{thm:class2}. It lists the types of ternary
sharp operations obtained by superposing a ternary sharp operation of an arbitrary
type (columns in Table~\ref{tab:types}) with any of the three cyclic permutations of the
three ternary projections (rows in Table~\ref{tab:types}).

\begin{table}[ht]
  \begin{center}
  \begin{tabular}{c|*{8}{c}}
    \makebox[1.5cm]{Permutation} 
    & \makebox[0.75cm]{Mj} & \makebox[0.75cm]{S1} & \makebox[0.75cm]{S2}
    & \makebox[0.75cm]{S3} & \makebox[0.75cm]{P1} & \makebox[0.75cm]{P2}
    & \makebox[0.75cm]{P3} & \makebox[0.75cm]{Mn}\\
    \hline
    $\proj^{(3)}_1,\proj^{(3)}_2,\proj^{(3)}_3$ & Mj & S1 & S2 & S3 & P1 & P2 & P3 & Mn\\
    $\proj^{(3)}_2,\proj^{(3)}_3,\proj^{(3)}_1$ & Mj & S2 & S3 & S1 & P3 & P1 & P2 & Mn\\
    $\proj^{(3)}_3,\proj^{(3)}_1,\proj^{(3)}_2$ & Mj & S3 & S1 & S2 & P2 & P3 & P1 & Mn\\
  \end{tabular}
  \end{center}
  \caption{Types of ternary sharp operations superposed with cyclic permutations of
  projections}\label{tab:types}
\end{table}

Note that taking a semiprojection $f$ of type $i$ and a Pixley operation $g$ of
type $i$, $f^\pi$ and $g^\pi$ can be of different types; e.g., if $f$ is a
semiprojection of type $1$ and $g$ is a Pixley operation of type $1$ and
$\pi=(2,3,1)$ then $f^\pi$ is a semiprojection of type $2$ and $g^\pi$ is a
Pixley operation of type $3$.

\begin{proof}[Proof of Theorem~\ref{thm:class2}]
  It suffices to consider the ternary case as the rest of the theorem follows from (the proof
  of) Theorem~\ref{thm:class1} and the fact that $W$ is a rigid core, which
  eliminates the first case of Theorem~\ref{thm:class1}.
    
  Let $W$ be a weighted clone containing a ternary
  positive weighting $\omega$ such that every $f\in\supp(\omega)$ is sharp. 
  (If some $f\in\supp(\omega)$ were not sharp then we could show, as in the proof 
  of Theorem~\ref{thm:class1}, that the case~(1) holds.)
  We denote by $C$ the support clone of $W$.
  We assume that none of the cases (2a), (2b), (2c), (3) (with $k=3$) of the
  theorem applies as we would be done in any of these cases. 

  By Lemma~\ref{lem:perm}, we can assume that $\omega$ assigns
  weight $-1$ to each of the three ternary projections and thus
  \begin{equation}\label{eq:sum3} \sum_{f\in\supp(\omega)}\omega(f)\ =\ 3\,. \end{equation}
 
  Let $P_i\subseteq C$ be the Pixley operations of type
  $i\in\{1,2,3\}$ from $C$.   
  Since $C$ is a clone we have 
  \begin{equation}\label{eq:Pcard} |P_1|\ =\ |P_2|\ =\ |P_3|\,. \end{equation}

  By Lemma~\ref{lem:perm}, we have 
  for any three related Pixley operations $p_1\in P_1$, $p_2\in P_2$, and $p_3\in P_3$,
  \begin{equation}\label{eq:Prelated} \omega(p_1)\ =\ \omega(p_2)\ =\ \omega(p_3)\,, \end{equation}
  and
  \begin{equation}\label{eq:Pweights} 
  \sum_{p\in P_1} \omega(p)\ =\ \sum_{p\in P_2} \omega(p)\ =\ \sum_{p\in P_3}
  \omega(p)\,. \end{equation} 
 
  We set $P=P_1\cup P_2\cup P_3$ to be the set of all Pixley operations from
  $C$ and $w(P)=\sum_{p\in P}\omega(p)$. 

  By Lemma~\ref{lem:perm}, 
  the same holds for the three types of ternary semiprojections. 
  In particular, we denote by $S_i\subseteq C$ the operations from $C$ that are semiprojections of type $i\in\{1,2,3\}$. 
  Since $C$ is a clone, we have
  \begin{equation}\label{eq:Scard} 
  |S_1|\ =\ |S_2|\ =\ |S_3|\,. 
  \end{equation}
  For any three related
  semiprojections $s_1\in S_1$, $s_2\in S_2$, and $s_3\in S_3$, 
  \begin{equation}\label{eq:Srelated}
  \omega(s_1)\ =\ \omega(s_2)\ =\ \omega(s_3)\,,
  \end{equation}
  and
  \begin{equation}\label{eq:Sweights} 
  \sum_{s\in S_1} \omega(s)\ =\ \sum_{s\in S_2} \omega(s)\ =\ \sum_{s\in S_3} \omega(s)\,. 
  \end{equation} 

  We set $S=S_1\cup S_2\cup S_3$ to be the set of all semiprojections from $C$
  and $w(S)=\sum_{s\in S}\omega(s)$.

  To simplify the presentation, we use the same notation for weightings, 
  index sets, etc. in the following three steps since the steps are similar (but
  independent). Thus, for example, when one reads $J$ in Step II it refers
  to $J$ defined in Step II and not in Step I.

  \textbf{Step I}: Eliminating Pixley operations.

  We now show how to eliminate Pixley operations if needed, that is,
  assume $w(P)>0$ and thus some operations from $P$ are assigned positive
  weight.

  First assume that $\omega$ assigns positive weight to only Pixley operations,
  that is, $\supp(\omega)\subseteq P$. Hence $w(P)=3$. 
  Take arbitrary $p_1,p'_1\in P_1$, $p_2,p'_2\in P_2$, and $p_3,p'_3\in P_3$.
  The following claims can be verified from the definitions:
  $p_1[\proj^{(3)}_1,\proj^{(3)}_2,p'_1]$ is a majority operation,
  $p_2[\proj^{(3)}_1,\proj^{(3)}_2,p'_1]=\proj^{(3)}_1$, and
  $p_3[\proj^{(3)}_1,\proj^{(3)}_2,p'_1]=\proj^{(3)}_2$.  Consequently,
  $\omega[\proj^{(3)}_1,\proj^{(3)}_2,p'_1]$ assigns weight $-1$ to $p'_1$, $+1$
  to majority operations, and $0$ otherwise. Similarly,
  $\omega[\proj^{(3)}_1,p'_2,\proj^{(3)}_3]$ assigns weight $-1$ to $p'_2$, $+1$
  to majority operations, and $0$ otherwise. Finally, 
  $\omega[p'_3,\proj^{(3)}_2,\proj^{(3)}_3]$ assigns weight $-1$ to $p'_3$, $+1$
  to majority operations, and $0$ otherwise. Overall, the weighting
  \begin{equation}\label{eq:pixleyonly}
  \mu\ =\ \omega+\sum_{p\in P_1}\omega(p)\omega[\proj^{(3)}_1,\proj^{(3)}_2,p]
  +\sum_{p\in P_2}\omega(p)\omega[\proj^{(3)}_1,p,\proj^{(3)}_3]
  +\sum_{p\in P_3}\omega(p)\omega[p,\proj^{(3)}_2,\proj^{(3)}_3]
  \end{equation}
  assigns weight $-1$ to each of the three ternary projections and weight $3$ to
  majority operations. By Lemma~\ref{lemma:sumofsuperpos}, the intermediate
  superpositions in~\eqref{eq:pixleyonly} can be improper as long as the
  resulting weighting $\mu$ is indeed a weighting. Thus $\mu\in W$ and case~(2a)
  of the theorem holds.

  Let assume that $\supp(\omega)\not\subseteq P$, that is, $\omega$ assigns positive weight not only to Pixley
  operations.
  Let $J=\{(p_1,p_2,p_3)\in P_1\times P_2\times P_3\:|\:\mbox{$p_1, p_2, p_3$ are related}\}$
  and $\bar{J}=\{(\proj^{(3)}_1,\proj^{(3)}_2,\proj^{(3)}_3)\}\cup J$.

  We consider the following
  linear system: for all Pixley operations $p\in P$,
  
  \begin{equation} \label{eq:pixleyzero} \sum_{(f,g,h) \in \bar{J}}x_{f,g,h}\omega[f,g,h](p)\ =\ 0\,. \end{equation}
  
  By Gordan's Theorem,~\eqref{eq:pixleyzero} has a nonzero nonnegative solution
  if and only if the following system of strict
  inequalities is unsatisfiable: for all $(f,g,h) \in \bar{J}$,
  
  \begin{equation} \label{eq:pixleygordan} \sum_{p \in P}y_p\omega[f,g,h](p)\ >\
  0\,. \end{equation}

  Consider the case $\tuple{f,g,h}=\tuple{\proj^{(3)}_1,\proj^{(3)}_2,\proj^{(3)}_3}$.
  Then $\omega[f,g,h]=\omega$ and thus $\omega[f,g,h](p)=\omega(p)>0$ for all
  $p\in \supp(\omega)\cap P$,
  by the definition of $P$. Moreover, by~\eqref{eq:Pweights} 
  \begin{equation}\label{eq:split}
  \sum_{p\in P}y_p\omega(p) \ =\ \sum_{p_1\in
  P_1}(y_{p_1}+y_{p_2}+y_{p_3})\omega(p_1)\,,
  \end{equation} 
  where we denoted (with a slight abuse of notation) by $p_2$ and $p_3$ the related operations of $p_1$.
 
  For $\tuple{f,g,h} = \tuple{\proj^{(3)}_1,\proj^{(3)}_2,\proj^{(3)}_3}$ 
  the LHS of~\eqref{eq:pixleygordan} is equal to~\eqref{eq:split}. Thus,
  for~\eqref{eq:pixleygordan} to hold in this case
  we must have at least one
  triple of related operations $\tuple{p_1,p_2,p_3} \in J$ with $y_{p_1} + y_{p_2} + y_{p_3} > 0$.

  Suppose $\tuple{p_1,p_2,p_3}\in J$ is chosen
  to maximise
  $y_{p_1} + y_{p_2} + y_{p_3}$. 
  The left-hand side of \eqref{eq:pixleygordan} when
  $\tuple{f,g,h} = \tuple{p_1,p_2,p_3}$ is equal to
  \begin{equation}\label{eq:abc}
    - y_{p_1} - y_{p_2} - y_{p_3} + \sum_{s \in S}y_{s[p_1,p_2,p_3]}\omega(s)\,,
  \end{equation}
  since $o[p_1,p_2,p_3]$ is equal to a Pixley operation only when $o$ is one
  of the three projections, which give the first three terms in~\eqref{eq:abc},
  or a semiprojection, which gives the last term of~\eqref{eq:abc}, the sum over $S$.

  We have
  \begin{equation}\label{eq:S3}
  \begin{split}
    \sum_{s \in S}y_{s[p_1,p_2,p_3]}\omega(s)\ =\ &
    \sum_{s_1\in S_1}y_{s_1[p_1,p_2,p_3]}\omega(s_1) +
    \sum_{s_2\in S_2}y_{s_2[p_1,p_2,p_3]}\omega(s_2) + \sum_{s_3\in S_3}y_{s_3[p_1,p_2,p_3]}\omega(s_3) \\
    =\ & \sum_{\substack{\tuple{s_1,s_2,s_3}\in S_1\times S_2\times S_3\\s_1,s_2,s_3\mbox{ \scriptsize related}}}(y_{s_1[p_1,p_2,p_3]}+y_{s_2[p_1,p_2,p_3]}+y_{s_3[p_1,p_2,p_3]})\omega(s_1) \\
    \leq\ & \frac{\omega(S)}{3}(y_{p_1}+y_{p_2}+y_{p_3}),\\
  \end{split}
  \end{equation}
  where the first equality follows from the definition of $S$; the second
  equality follows from fact that $s[p_1,p_2,p_3]$ is a Pixley operation of type $i$
  given $s$ is a semiprojection of type $i$, where $i\in\{1,2,3\}$, and hence 
  $\tuple{s_1[p_1,p_2,p_3],s_2[p_1,p_2,p_3],s_3[p_1,p_2,p_3]}$ is a triple of
  related Pixley operations given that $\tuple{s_1,s_2,s_3}$ is a triple of
  related semiprojections;  and the
  last inequality follows from the definition of $w(S)$ and the choice of
  $(p_1,p_2,p_3)$.

  Combining~\eqref{eq:abc} and~\eqref{eq:S3}, we have
  \begin{equation}
  - y_{p_1} - y_{p_2} - y_{p_3} + \sum_{s \in S}y_{s[p_1,p_2,p_3]}\omega(s)
    \ \leq\ 
    - y_{p_1} - y_{p_2} - y_{p_3} + \frac{w(S)}{3}(y_{p_1}+y_{p_2}+y_{p_3})
    \ <\ 0\,,
  \end{equation}

  where the last strict inequality follows from $w(S)<3$ since $w(P)>0$
  and~\eqref{eq:sum3}.

  Hence, \eqref{eq:pixleygordan} is unsatisfiable and, by Gordan's Theorem,
  \eqref{eq:pixleyzero} must have a nonzero nonnegative solution $x^*$. 
  We finish Step~I by using $x^*$ to prove the existence of a
  weighting in $W$ that assigns zero weight to all Pixley operations.
  
  Let
  \begin{equation} \label{eq:result} \mu'=\sum_{(f,g,h) \in \bar{J}}x^*_{f,g,h}\omega[f,g,h]\,, 
  \end{equation}
  be a weighted sum of superpositions of $\omega$. 
  
  By the choice of $x^*$, $\mu'$ assigns zero weight to all Pixley operations. 
  From the definition of $\bar{J}$,
  \begin{equation}\label{eq:mu}
   \mu'\ =\
   x^*_{\proj^{(3)}_1,\proj^{(3)}_2,\proj^{(3)}_3}\omega[\proj^{(3)}_1,\proj^{(3)}_2,\proj^{(3)}_3]+\sum_{(f,g,h)\in
   J}x^*_{f,g,h}\omega[f,g,h]\,.
  \end{equation}

  We know that $x^*$ is nonzero and nonnegative. Note that
  $x^*_{\proj^{(3)}_1,\proj^{(3)}_2,\proj^{(3)}_3}>0$ and $x^*_{f,g,h}=0$ for
  all $(f,g,h)\in J$ would contradict that $\mu'$ assigns zero weight to all
  Pixley operations. 

  Let
  $(f,g,h)\in J$ and $i\in\{1,2,3\}$. Then $o[f,g,h]$ is a Pixley operation of
  type $i$ if $o$ is a semiprojection of type $i$, $o[f,g,h]$ is a
  semiprojection of type $i$ if $o$ is a Pixley operation of type $i$,
  $o[f,g,h]$ is a majority operation if $o$ is a minority operation, and finally
  $o[f,g,h]$ is a minority operation if $o$ is a majority operation. 
  It follows that $x^*_{\proj^{(3)}_1,\proj^{(3)}_2,\proj^{(3)}_3}=0$ would
  contradict that $\mu'$ assigns zero weight all Pixley operations since Pixley
  operations of type $i$ would be assigned negative weight since $w(S)-3<0$ as
  $w(P)>0$. Thus $x^*_{\proj^{(3)}_1,\proj^{(3)}_2,\proj^{(3)}_3}>0$ and
  $x^*_{f,g,h}>0$ for at least one $(f,g,h)\in J$. Consequently, $\mu'$ is a
  nonzero weighting that assigns zero weight to all Pixley operations. By
  Lemma~\ref{lemma:sumofsuperpos}, $\mu'\in W$.

  \textbf{Step II}: Eliminating semiprojections.

  We now show how to eliminate semiprojections if needed. If $\mu'$ obtained in
  Step I assigns positive weight to only semiprojections then we are in case
  (3) (with $k$=3) of the theorem.  Thus let assume that $\mu'$ assigns positive
  weight to semiprojections and at least one majority or minority operation. 
  
  By Lemma~\ref{lem:perm}, we can
  assume the existence of a ternary weighting $\omega\in W$ which assigns weight -1 to all three projections (and thus
  assigns total positive weight 3). We will use the same notation for $S$,
  $S_1$, $S_2$, and $S_3$ as before although note that the weighting $\omega$ is
  now different from the one we had before and throughout Step I.
  We have $w(S)<3$
  and~\eqref{eq:Srelated},~\eqref{eq:Scard},~\eqref{eq:Sweights} still hold for
  $\omega$. 

  We again use Gordan's Theorem to show that there exists a nonzero ternary
  weighting in $W$ that assigns positive weight to majority and minority
  operations only. 

  Let $J=\{(s_1,s_2,s_3)\in S_1\times S_2\times S_3\:|\:\mbox{$s_1, s_2, s_3$ are related}\}$
  and $\bar{J}=\{(\proj^{(3)}_1,\proj^{(3)}_2,\proj^{(3)}_3)\}\cup J$.

  We consider the following linear system: for all semiprojections $s\in S$,
  
  \begin{equation} \label{eq:semiprojzero} \sum_{\tuple{f,g,h}\in
  \bar{J}}x_{f,g,h}\omega[f,g,h](s)\ =\ 0\,. \end{equation} 

  By Gordan's Theorem,~\eqref{eq:semiprojzero} has a nonzero nonnegative solution
  if and only if the following system of strict
  inequalities is unsatisfiable: for all $(f,g,h) \in \bar{J}$,
  
  \begin{equation} \label{eq:semiprojgordan} 
  \sum_{s \in S}y_s\omega[f,g,h](s)\ >\ 0\,. 
  \end{equation} 

  As in Step I, we can argue that~\eqref{eq:semiprojgordan} is unsatisfiable.
  Consider the case $\tuple{f,g,h}=\tuple{\proj^{(3)}_1,\proj^{(3)}_2,\proj^{(3)}_3}$.
  Then $\omega[f,g,h]=\omega$ and thus $\omega[f,g,h](s)=\omega(s)>0$ for all
  $s\in\supp(\omega)\cap S$,
  by the definition of $S$. Moreover, by~\eqref{eq:Sweights} 
  \begin{equation}\label{eq:Ssplit}
  \sum_{s\in S}\omega(s)y_s \ =\ \sum_{s_1\in
  S_1}(y_{s_1}+y_{s_2}+y_{s_3})\omega(s_1)\,,
  \end{equation} 
  where we denoted by $s_2$ and $s_3$ the related operations of $s_1$.

  Therefore, for \eqref{eq:semiprojgordan} to hold when $\tuple{f,g,h} =
  \tuple{\proj^{(3)}_1,\proj^{(3)}_2,\proj^{(3)}_3}$ we must have at least one
  triple of related semiprojections $\tuple{s_1,s_2,s_3} \in J$
  with $y_{s_1} + y_{s_2} + y_{s_3} > 0$. 
  
  Suppose $\tuple{s_1,s_2,s_3}\in J$ is chosen
  to maximise
  $y_{s_1} + y_{s_2} + y_{s_3}$. 
  The left-hand side of \eqref{eq:semiprojgordan} when
  $\tuple{f,g,h} = \tuple{s_1,s_2,s_3}$ is equal to
  \begin{equation}\label{eq:Sabc}
    - y_{s_1} - y_{s_2} - y_{s_3} + \sum_{s \in S}y_{s[s_1,s_2,s_3]}\omega(s)\,,
  \end{equation}
  since $o[s_1,s_2,s_3]$ is equal to a semiprojection only when $o$ is one
  of the three projections, which give the first three terms in~\eqref{eq:Sabc},
  or a semiprojection, which gives the last term of~\eqref{eq:Sabc}, the sum over $S$.

  We have
  \begin{equation}\label{eq:SS3}
  \begin{split}
    \sum_{s \in S}y_{s[s_1,s_2,s_3]}\omega(s)\ =\ &
    \sum_{s_1\in S_1}y_{s_1[s_1,s_2,s_3]}\omega(s_1) +
    \sum_{s_2\in S_2}y_{s_2[s_1,s_2,s_3]}\omega(s_2) + \sum_{s_3\in S_3}y_{s_3[s_1,s_2,s_3]}\omega(s_3) \\
    =\ & \sum_{\substack{\tuple{s_1,s_2,s_3}\in S_1\times S_2\times
    S_3\\s_1,s_2,s_3\mbox{\scriptsize related}}}(y_{s_1[s_1,s_2,s_3]}+y_{s_2[s_1,s_2,s_3]}+y_{s_3[s_1,s_2,s_3]})\omega(s_1) \\
    \leq\ & \frac{\omega(S)}{3}(y_{s_1}+y_{s_2}+y_{s_3}),\\
  \end{split}
  \end{equation}
  where the first equality follows from the definition of $S$; the second
  equality follows from fact that $s[s_1,s_2,s_3]$ is a semiprojection of type $i$
  given $s$ is a semiprojection of type $i$, where $i\in\{1,2,3\}$, and hence 
  $\tuple{s_1[s_1,s_2,s_3],s_2[s_1,s_2,s_3],s_3[s_1,s_2,s_3]}$ is a triple of
  related semiprojections given that $\tuple{s_1,s_2,s_3}$ is a triple of
  related semiprojections;  and the
  last inequality follows from the definition of $w(S)$ and the choice of
  $(s_1,s_2,s_3)$.

  Combining~\eqref{eq:Sabc} and~\eqref{eq:SS3}, we have
  \begin{equation}
    - y_{s_1} - y_{s_2} - y_{s_3} + \sum_{s \in S}y_{s[s_1,s_2,s_3]}\omega(s)
    \ \leq\ 
    - y_{s_1} - y_{s_2} - y_{s_3} + \frac{w(S)}{3}(y_{s_1}+y_{s_2}+y_{s_3})
    \ <\ 0\,,
  \end{equation}
 where the last strict inequality follows from $w(S)<3$.

  Hence, \eqref{eq:semiprojgordan} is unsatisfiable and, by Gordan's Theorem,
  \eqref{eq:semiprojzero} must have a nonzero nonnegative solution $x^*$. 
  We finish Step~II by using $x^*$ to prove the existence of a
  weighting in $W$ that assigns zero weight to all semiprojections.
  
  Let
  \begin{equation} \label{eq:Sresult} \mu=\sum_{(f,g,h) \in J}\omega[f,g,h]x^*_{f,g,h}\,, 
  \end{equation}
  be a weighted sum of superpositions of $\omega$. By the choice of $x^*$, $\mu$ 
  assigns zero weight to all semiprojections. 

  From the definition of $\bar{J}$,
  \begin{equation}\label{eq:Smu}
   \mu\ =\
   x^*_{\proj^{(3)}_1,\proj^{(3)}_2,\proj^{(3)}_3}\omega[\proj^{(3)}_1,\proj^{(3)}_2,\proj^{(3)}_3]+\sum_{(f,g,h) \in J}x^*_{f,g,h}\omega[f,g,h]\,.
  \end{equation}

  We know that $x^*$ is nonzero and nonnegative. Note that
  $x^*_{\proj^{(3)}_1,\proj^{(3)}_2,\proj^{(3)}_3}>0$ and $x^*_{f,g,h}=0$ for
  all $(f,g,h)\in J$ would contradict that $\mu'$ assigns zero weight to all
  semiprojections. 

  Let $(s_1,s_2,s_3)\in J$ and $i\in\{1,2,3\}$. Then $o[s_1,s_2,s_3]$ is a
  semiprojection of type $i$ if $o$ is a semiprojection of type $i$,
  $o[s_1,s_2,s_3]$ is a majority operation if $o$ is a majority operation, and
  finally $o[s_1,s_2,s_3]$ is a minority operation if $o$ is a minority
  operation. (Note that we do not need to consider the case when $f$ is a Pixley
  operation as $\omega$ assigns zero weight to all  Pixley operations.) 
  It follows that $x^*_{\proj^{(3)}_1,\proj^{(3)}_2,\proj^{(3)}_3}=0$ would
  contradict that $\mu$ assigns zero weight all semiprojections since
  semiprojections of type $i$ would be assigned negative weight since $w(S)-3<0$
  as $\omega$ assigns positive weight to some majority or minority (or possibly
  both) operations. 
  Thus $x^*_{\proj^{(3)}_1,\proj^{(3)}_2,\proj^{(3)}_3}>0$ and $x^*_{f,g,h}$ for
  at least one $(f,g,h)\in J$. Consequently, $\mu$ is a nonzero weighting that
  assigns zero weight to all semiprojections. By
  Lemma~\ref{lemma:sumofsuperpos}, $\mu\in W$.

  \textbf{Step III}: Majority and minority operations.

  In Steps I and II, we have shown that any weighted clone $W$ with a positive ternary weighting
  contains a ternary weighting that assigns nonzero weight to semiprojections
  alone (case~(3) of the theorem with $k=3$), or a mix of majority and minority operations. Finally, we will show
  that if $W$ contains a weighting $\omega$
  that assigns weight to majority and minority operations alone then $W$ also
  contains a weighting of one of the three types described in
  cases~(2a),~(2b), and~(2c) of the theorem.
  
  Let $M_1$ and $M_2$ denote the sets of majority and minority operations in the
  support of $\omega$. Suppose that $\omega$ assigns total weight $2+a$ to $M_1$
  and total weight $1-a$ to $M_2$ for some $a > 0$. 
  For each $f \in M_2$, we define $\mu_f =
  \omega[\proj^{(3)}_1,\proj^{(3)}_1,f]$, so $\mu_f$ assigns weight $a$ to
  $\proj^{(3)}_1$ and weight $-a$ to $f$. Note that $\mu_f$ is not a proper
  weighting, since $a > 0$. We obtain a weighting $\mu \in W$
  which assigns positive weight to majority operations only as follows:

  \begin{equation}
    \mu = \omega + \sum_{f \in M_2}\frac{\omega(f)}{a}\mu_f\,.
  \end{equation}

  Similarly, suppose that $\omega$ assigns total weight $2-a$ to $M_1$ and total
  weight $1+a$ to $M_2$ for some $a>0$.
  For each $f\in M_1$, we define $\mu_f=\omega[\proj^{(3)}_1,f,f]$, so $\mu_f$
  assigns weight $a$ to $\proj^{(3)}_1$ and weight $-a$ to $f$. We obtain a
  weighting $\mu \in W$ which assigns positive weight to minority operations
  only as follows:

  \begin{equation}
    \mu = \omega + \sum_{f \in M_1}\frac{\omega(f)}{a}\mu_f\,.
  \end{equation}
  In both cases $\mu\in W$ by Lemma~\ref{lemma:sumofsuperpos}.
\end{proof}

\section{Conclusions}

We have presented new results on the structure of weighted clones that delimit
the possibilities for tractable valued constraint languages. In order to
establish our results, we have presented a novel technique for ruling out
certain types of operations from the support of a given weighting. The method
considers certain extreme cases of the dual of the linear program that
demonstrates the existence of a weighted sum of superpositions that assigns zero
weight to the forbidden operations. We believe that our results and techniques
will prove useful in further studies of the structure of weighted clones.
However, understanding the structure of weighted clones appears a difficult
problem in
general. For instance, whilst the computational complexity of finite-valued
constraint languages is well understood~\cite{tz13:stoc}, the structure of the
corresponding weighted clones is not, as discussed in Section~\ref{sec:wclones}.

In recent work on the tractability of valued constraint languages, it has been
shown that a necessary condition for tractability is the existence of a
\emph{cyclic} weighted polymorphism~\cite{Kozik15:icalp}.\footnote{A $k$-ary
operation if \emph{cyclic} if $f(x_1,x_2,\ldots,x_k)=f(x_2,\ldots,x_k,x_1)$ for
every $x_1,\ldots,x_k$. A weighting $\omega$ is cyclic if every operation
$f\in\supp(\omega)$ is cyclic.} Moreover, it has been also shown that, under the
assumption of the dichotomy conjecture of Feder and Vardi for the decision
problem, this condition is
also sufficient~\cite{Kolmogorov15:general-valued}.

\section*{Acknowledgements}
The authors are grateful to P\'aid\'i Creed and Peter Fulla for valuable discussions.


\newcommand{\noopsort}[1]{}

\end{document}